\documentclass[preprint]{imsartNoLabel}
\usepackage{amsthm,amsmath,amssymb}
\usepackage{subfig}
\usepackage{graphicx}
\usepackage{natbib}
\RequirePackage[colorlinks,citecolor=blue,urlcolor=blue]{hyperref}

\startlocaldefs
\DeclareMathOperator*{\argmin}{\arg\!\min} 
\newtheorem{theorem}{Theorem}
\newtheorem{lemma}[theorem]{Lemma}
\newtheorem{proposition}[theorem]{Proposition}

\newenvironment{example}[1][Example]{\begin{trivlist}
                                     \item[\hskip \labelsep {\bfseries #1}]}{\end{trivlist}}
\theoremstyle{definition}
\newtheorem{remark}{Remark}[section]

\newcommand{\be}{\beta}
\newcommand{\lamVec}{\boldsymbol{\lambda}}
\newcommand{\ep}{\varepsilon}
\newcommand{\gen}{\mathcal{G}}
\newcommand{\hypSq}{\mathcal{L}}
\newcommand{\mono}{\mathcal{M}}

\newcommand{\sigmaSqj}{\frac{\ep^2}{\Delta_{nj}}}

\newcommand{\B}{\mathcal{B}}

\newcommand{\IAj}{\mathbf{1}_{A_j}}
\newcommand{\IAjc}{\mathbf{1}_{A_j^c}}
\newcommand{\tZ}{\tilde Z}
\newcommand{\parm}{T^2}
\newcommand{\lamMax}{\lambda_{\text{max}}}
\newcommand{\IAjt}{\mathbf{1}_{\tilde A_j}}
\newcommand{\avgY}{\overline{Y}_n}
\newcommand{\avgK}{\overline{K}_n}
\newcommand{\avgX}{\overline{X}_n}
\newcommand{\avgD}{\overline{D}_n}
\newcommand{\avgE}{\overline{\mathcal{E}}}
\newcommand{\avgB}{\overline{\mathbf{B}}_n}
\newcommand{\thetaS}{\theta^{\text{smooth}}}
\newcommand{\thetaP}{\theta^{\text{peaked}}}
\newcommand{\hthetaRidge}{\hat\theta_{\text{ridge}}}
\endlocaldefs

\begin{document}

\begin{frontmatter}

\title{Efficient Estimators for Sequential and Resolution-Limited Inverse Problems}
\runtitle{Estimators for Sequential Inverse Problems}


\author{\fnms{Darren} \snm{Homrighausen}\ead[label=e1]{dhomrigh@andrew.cmu.edu}}
\address{Department of Statistics\\Colorado State University\\
Fort Collins, CO 80523\\\printead{e1}}
\and
\author{\fnms{Christopher R.} \snm{Genovese}\ead[label=e2]{genovese@stat.cmu.edu}}
\address{Department of Statistics\\Carnegie Mellon University\\
Pittsburgh, PA 15223\\\printead{e2}}
\affiliation{Carnegie Mellon University}

\runauthor{Homrighausen and Genovese}

\begin{abstract}
A common problem in the sciences is that a signal of interest is observed only indirectly, 
through smooth functionals of the signal whose values are then obscured by noise.
In such \emph{inverse problems},
the functionals dampen or entirely eliminate some of the signal's interesting features.
This makes it difficult or even impossible to fully reconstruct
the signal, even without noise.
In this paper, we develop methods for handling 
\emph{sequences} of related inverse problems,
with the problems varying either systematically or randomly over time.
Such sequences often arise with automated data collection systems,
like the data pipelines of large astronomical instruments such as the
Large Synoptic Survey Telescope (LSST).
The LSST will observe each patch of the sky many times over its lifetime
under varying conditions.
A possible additional complication in these problems is that the observational
resolution is limited by the instrument, so that even with many repeated observations,
only an approximation of the underlying signal can be reconstructed.
We propose an efficient estimator for reconstructing a signal of interest
given a sequence of related, resolution-limited inverse problems.
We demonstrate our method's effectiveness in some representative examples and
provide theoretical support for its adoption.
\end{abstract}



\end{frontmatter}

\maketitle

\section{Introduction}
\label{sec:introduction}
In many applications, data about a signal of interest can only be indirectly gathered.
For instance, astronomical images from ground-based telescopes
are observed through the blurring caused by atmospheric turbulence;
Positron Emission Tomography (PET) scanners measure photon intensities averaged over lines;
and seismologists record the surface effects of earthquakes whose waves have been filtered by the Earth.
In these examples and other such \emph{inverse problems},
the basic measurements are smooth functionals of the signal
that dampen or entirely eliminate some of the signal's interesting features.
This makes it difficult, or sometimes impossible, to fully reconstruct the signal from noisy data.

Over the years, a number of methods have been developed for the recovery of a signal
under inverse problems.  We cannot hope to provide a comprehensive list, but see
\citet{osullivan1986,wahba1990,donoho1995wvd,tenorio2001,candes2002,cavalier2002b}
and the references contained therein for an introduction and
\citet{cavalier2008} for a modern
review of the state of the field.  Also,
many disciplines have developed specific techniques for addressing
particular issues, such as Astronomy \citep{starck2002,vanDyk2006} and Tomography \citep{olafsson2005}.

However, the above cited work provides techniques and theory for situations in which 
an estimate of a signal is formed after only one observation.
In many fields, recent technological advances have made it possible to automate data-collection, 
enabling repeated observations of the signal over time.
While repeated observations can improve accuracy,
it often raises new challenges as
the inverse problems faced at different times can vary significantly.
For example, the
Large Synoptic Survey Telescope (LSST),
a multi-year, Earth-based survey of the entire sky,
will image space to an unprecedented depth and will catalog billions of astronomical objects.
The LSST will take long sequences of images at each patch of sky, about 3 degrees on a side.
In each sequence, the images will be separated in time by approximately 3--4 days.
Each image in each sequence is taken with different
blurring and distortion conditions.
Thus, the viewing process represents sequences of related but distinct inverse problems.
One scientific goal is to use these images to reconstruct the signal, which
in this case
is comprised of the underlying celestial structures,  as accurately
as possible.

Notationally, we consider the following problem.
We want to recover information about an unknown signal
$\theta \in \mathbb{R}^p$ from measurements of the form
\begin{equation}
\mathbf{Y}_i = K_i \theta + \ep \mathbf{W}_i,  \quad{\rm for} \quad i = 1,2,\ldots.
\label{eq:DWNP}
\end{equation}
Here, each $\mathbf{Y}_i$ is a measured signal, such as an audio recording or a (vectorized) image
represented as a $p \times 1$ vector.
Each \emph{forward operator} $K_i$ describes the measurement process and
the $\mathbf{W}_i$'s are independent, mean zero Gaussian $p$-vectors with variance-covariance
matrix $I_p$, the order $p$ identity.

The $K_i$ represent both
the damping of the signal present in an inverse problem and the necessary discretization
due to the resolution-limited nature of most observational devices, most commonly
through pixelization.  The $K_i$'s are a priori unknown and hence must be measured and estimated.
As any information about the $K_i$'s comes from the observational device itself,
any estimate of the $K_i$'s are resolution-limited as well.  Therefore, we represent the measurement
process $K_i$ as a $p\times p$ matrix.
This captures the idea that, in many problems, the resolution is fixed
by the instrument and does not change as more data is collected (that is,
as $n\to \infty$).

An early formal consideration of the sequential inverse problem
is found in the literature on developing loss-less analogue-to-digital 
conversion techniques.  The recovery of the orginal, analogue
signal is an inverse problem
as there is not a unique analogue signal corresponding to each digital signal. This
result is formalized in the quantity referred to as the Nyquist rate, or frequency
\citep[Chapter 3]{mallet2009}.
If the signal is instead sampled multiple times at different, carefully chosen 
sampling rates, \citet{berenstein1990} 
and \citet{casey1994} find conditions under which
the original signal can be reconstructed in a loss-less way. 
Note that, as opposed to our paper, these approaches deal with
only the case where $\epsilon = 0$ and the $K_i$, which correspond to
the sampling rate, can be chosen
by the experimenter.

Subsequently, the sequential inverse problem is considered in a 
series of articles, beginning with  \citet{piana1996}, in which two methods are introduced.
The first corresponds to Tikonov-Phillips (TP) regularization (known in statistics
as ridge regression) 
adapted to the sequential problem.  The second is an iterative method based on
Landwieber iterations (LI).  See \citet{bertero1998} 
for an overview of the Landwieber iterations method in inverse problems.
Though the above methods have been successfully implemented in the past,
most notably in the software package AIRY \citep{correia2002}, 
it has two shortcomings: the methods correspond to restrictive choices among all possible
estimators and they
offer no automated method for choosing the introduced tuning parameters.

\begin{remark}
The Tikonov-Phillips and Landwieber Iteration 
methods can  readily be
derived by the formalism developed in this paper 
(see equations \eqref{eq:LIestimator} and \eqref{eq:TPestimator}).   Therefore, 
we get for free a principled method for setting the tuning parameters, as
well as a suite of new estimators.
\end{remark}

The  goal of this paper is to develop and investigate a statistically efficient
estimator of $\theta$ from the sequence of resolution-limited inverse problems
introduced in equation \eqref{eq:DWNP}.
We require that any estimator must satisfy the following:
(i)~it leaves no user-defined tuning parameters
and 
(ii)~the estimator $\hat\theta_n$ based on an $n$-sequence can be efficiently
updated to produce the estimator $\hat\theta_{n+1}$ after observing $\mathbf{Y}_{n+1}$.
Both requirements are particularly important in applications like the
LSST, where it is inconvenient (or impossible) to access the entire past data stream with each
new observation and hence the data must be processed in near real time.

In Section \ref{sec:avNotEnough} we discuss two reasonable approaches to this problem
based on collapsing the sequence of operators $(K_i)$ into one summary operator,
in one case averaging the operators and in the other concatenating them.
We show that both approaches do not satisfy conditions (i) and (ii).

\begin{example}[Satellite Imaging:]
To fix ideas, we introduce a typical instance where the observations form a 
a sequence of resolution limited inverse problems.  During
satellite imaging operations, a location on Earth is imaged many times over the life span of the
satellite.  The quality of the recorded observations can be low 
and variable due to changing atmospheric and/or weather conditions.  
See the left column of Figure \ref{fig:satelliteExample} for a representative
panel of four such images taken of the White House and surrounding buildings.  
Note that the amount of blurring in each image $i$, corresponding to the forward operators $K_i$,
can be very different.  However, the pixelization induced by the observational
device is fixed over the sequence of images.

Our proposed estimator $\hat\theta_n$ 
takes these images and sequentially creates a new estimate of the unknown signal $\theta$
after each observation $Y_i$ (right column of Figure \ref{fig:satelliteExample}).  Each row of
Figure \ref{fig:satelliteExample} is a new observation and $\hat\theta_n$ after being updated 
with that observation.
Notice that the recovery is quite good, even after only a few images as input.
We emphasize that there are no choices
to be made by the data analyst: all tuning parameters are chosen in an automatic, data-dependent way.
\begin{figure}
  \centering
  \subfloat{\includegraphics*[width=2.5in,height=1.5in,trim=80 130 80 90,clip]{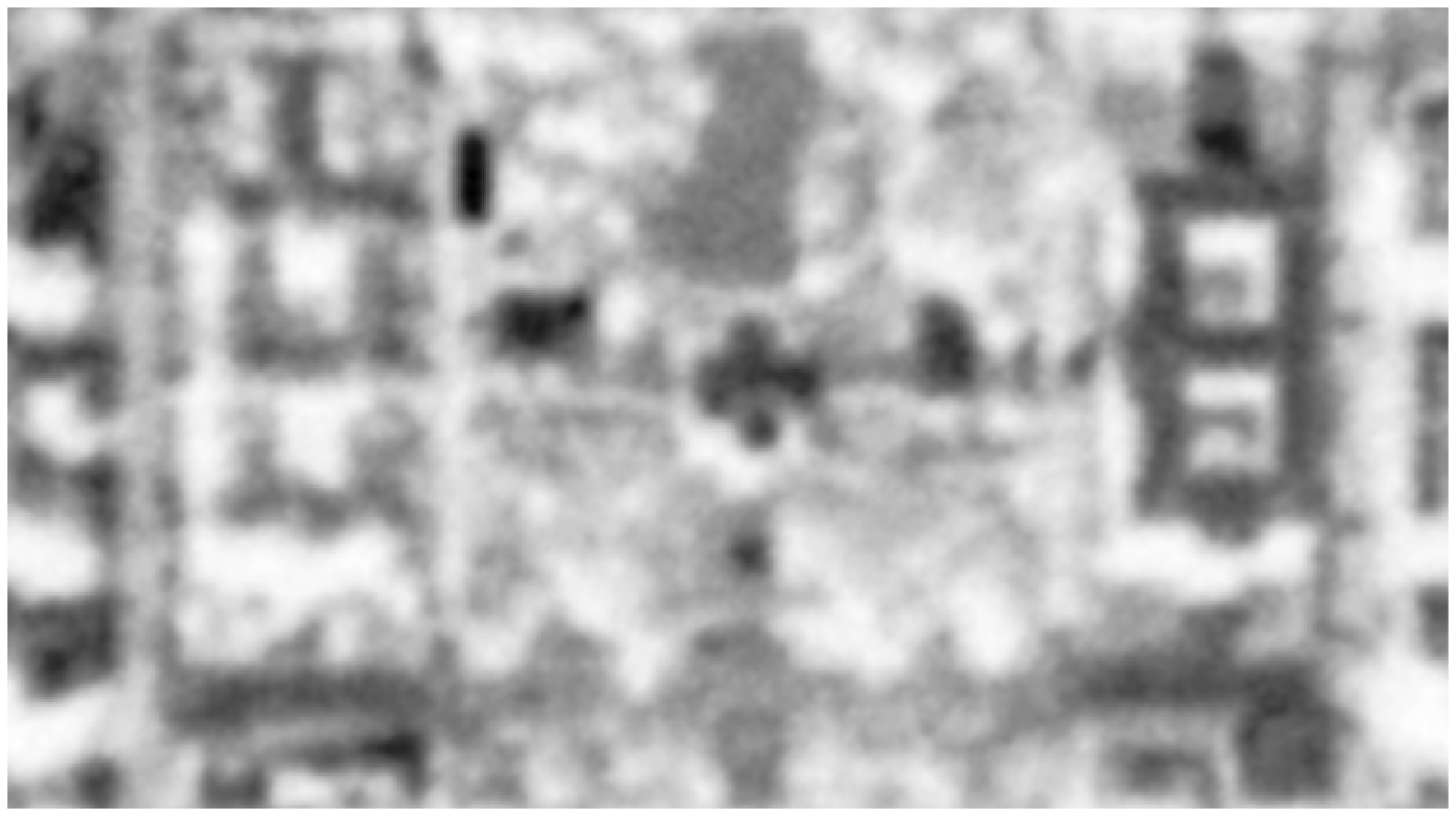} }
  \subfloat{\includegraphics*[width=2.5in,height=1.5in,trim=80 130 80 90,clip]{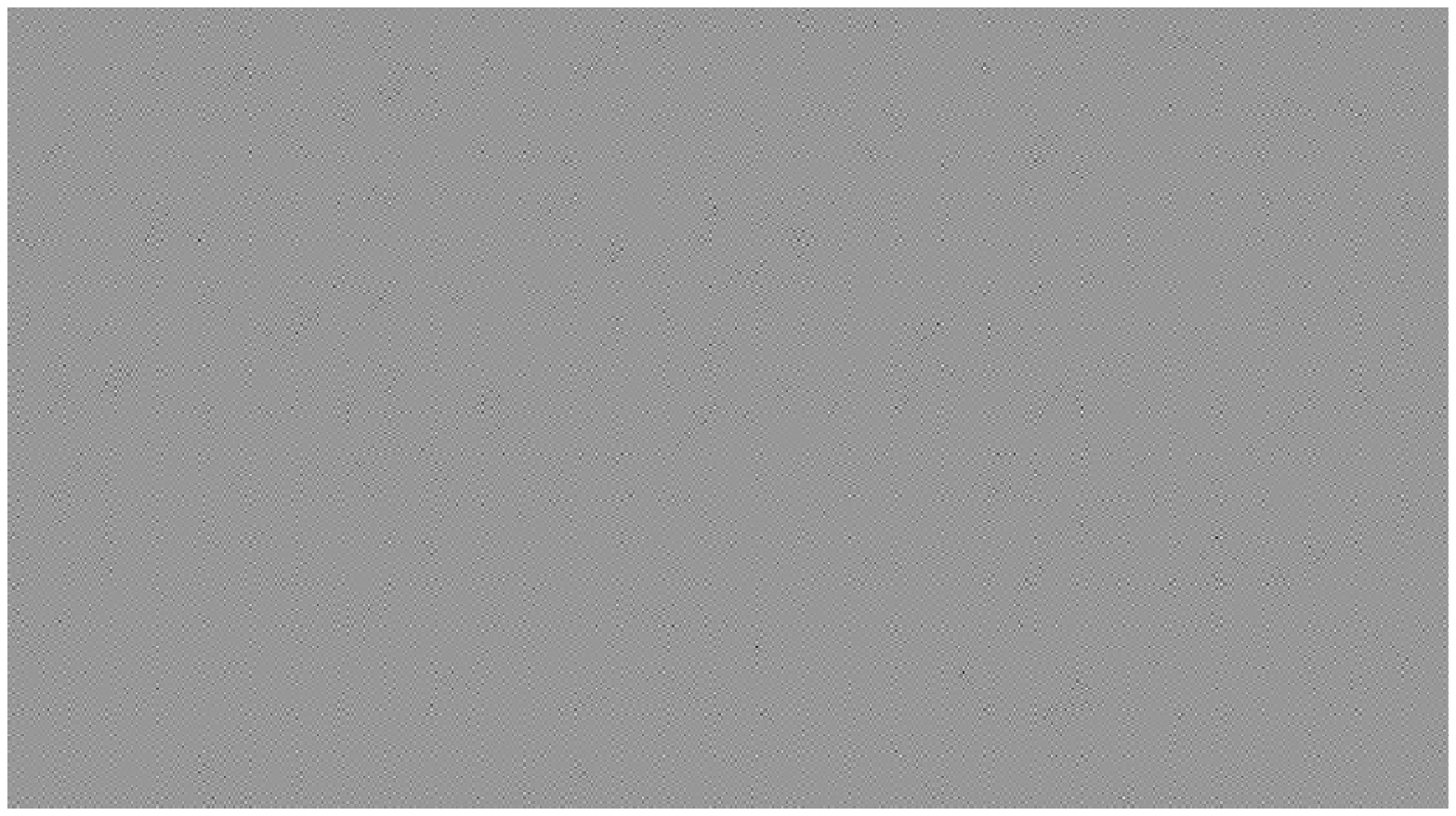} } \\
  \subfloat{\includegraphics*[width=2.5in,height=1.5in,trim=80 130 80 90,clip]{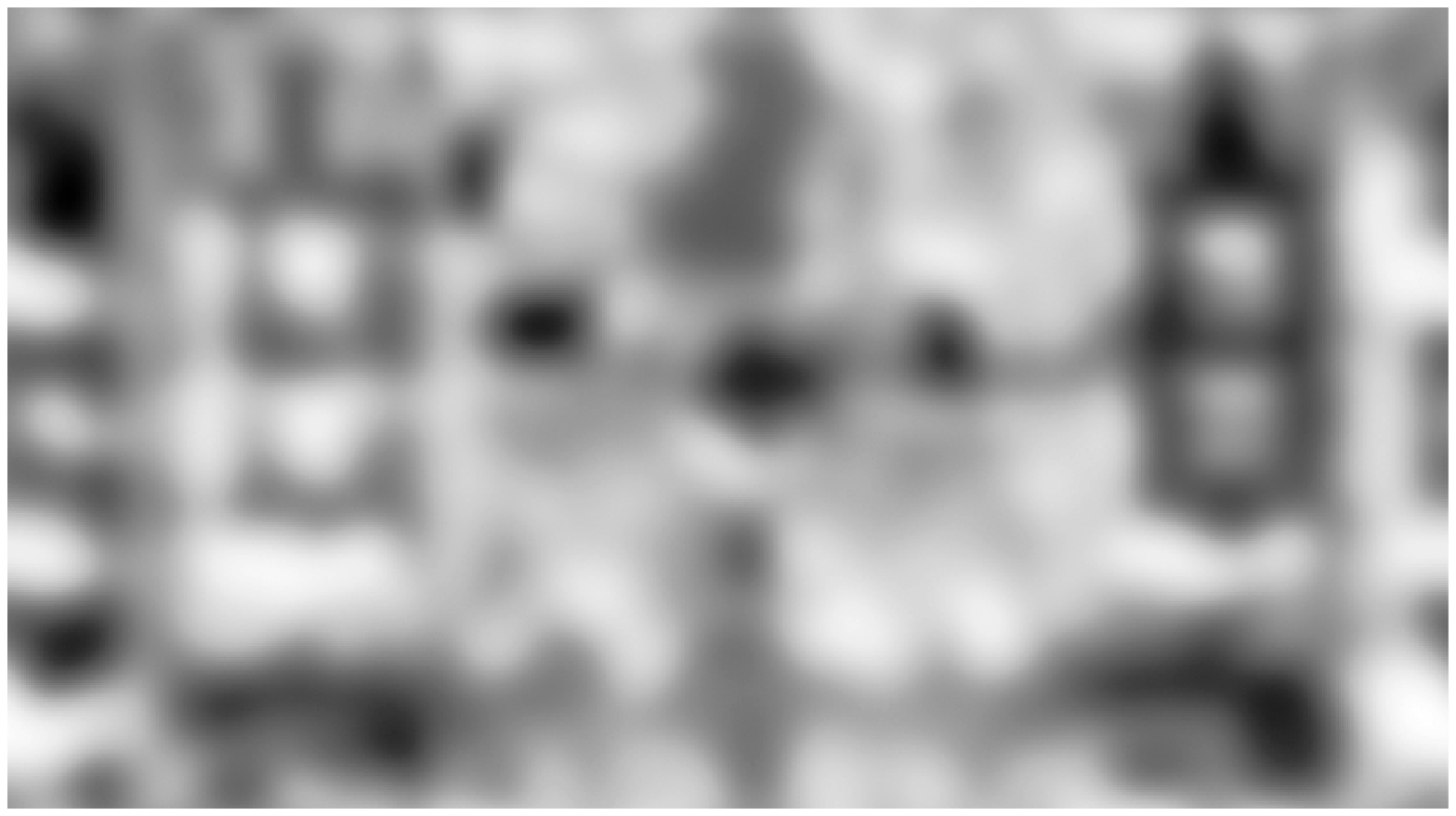} }
  \subfloat{\includegraphics*[width=2.5in,height=1.5in,trim=80 130 80 90,clip]{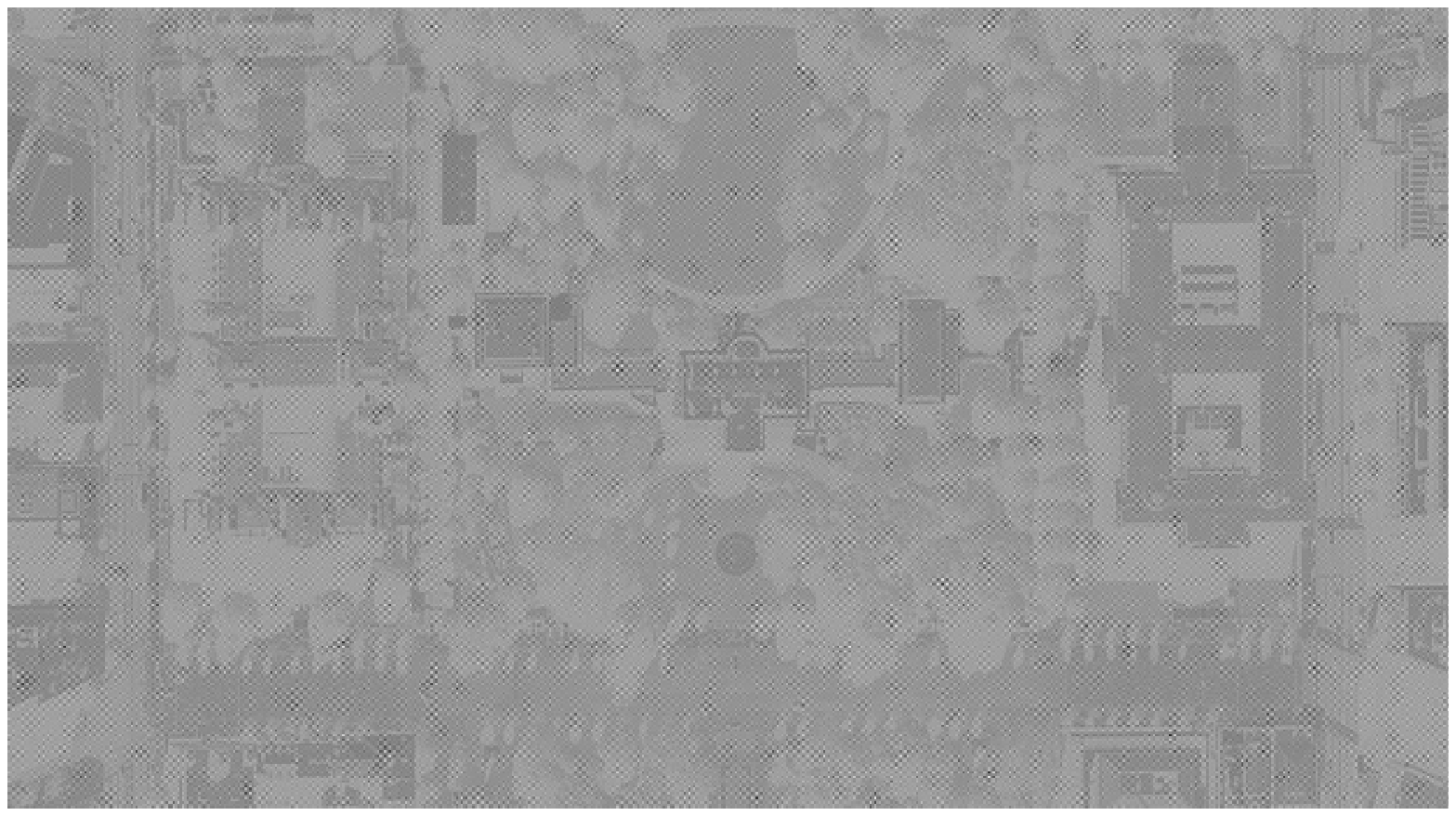} } \\
  \subfloat{\includegraphics*[width=2.5in,height=1.5in,trim=80 130 80 90,clip]{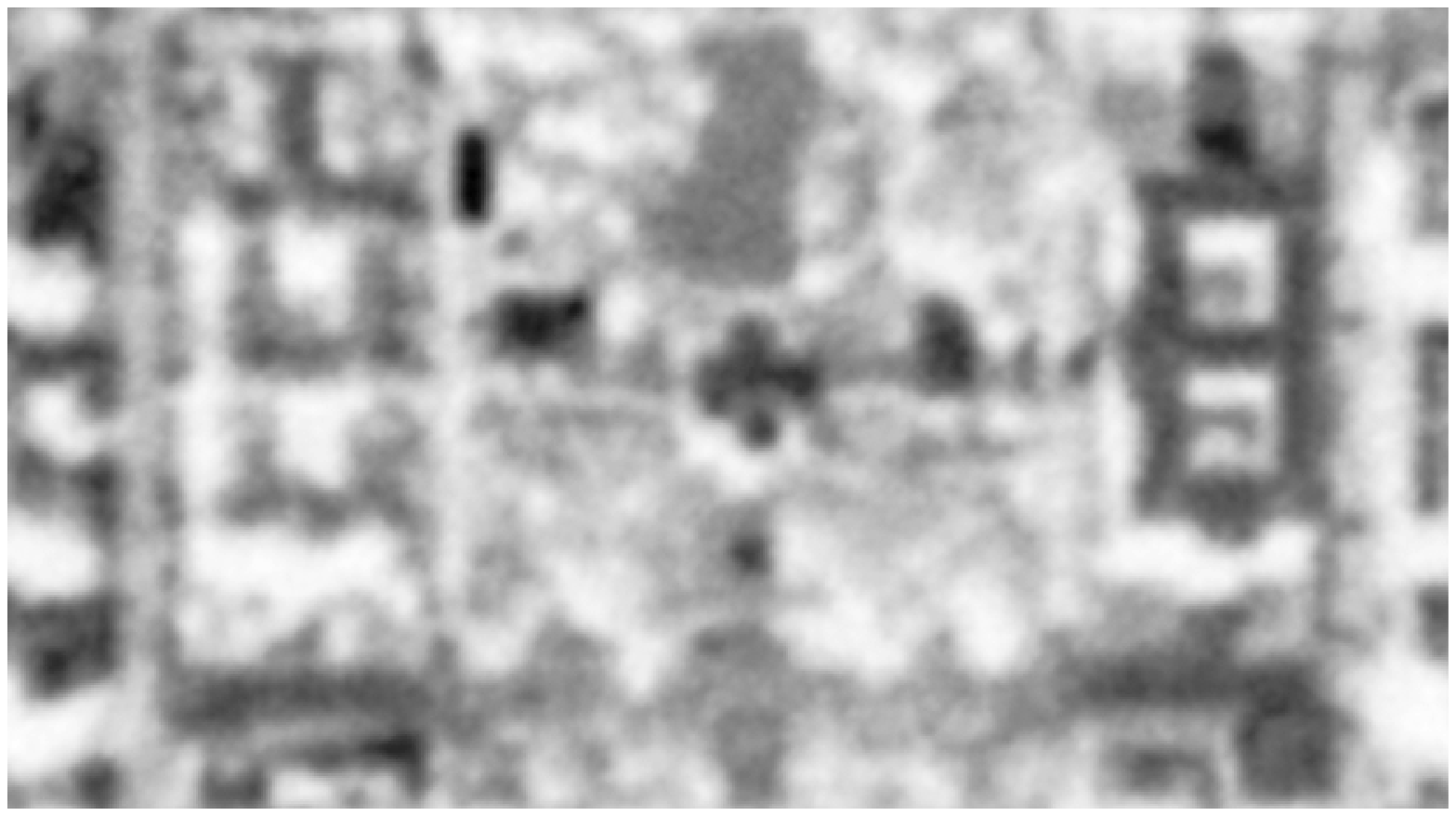}}
  \subfloat{\includegraphics*[width=2.5in,height=1.5in,trim=80 130 80 90,clip]{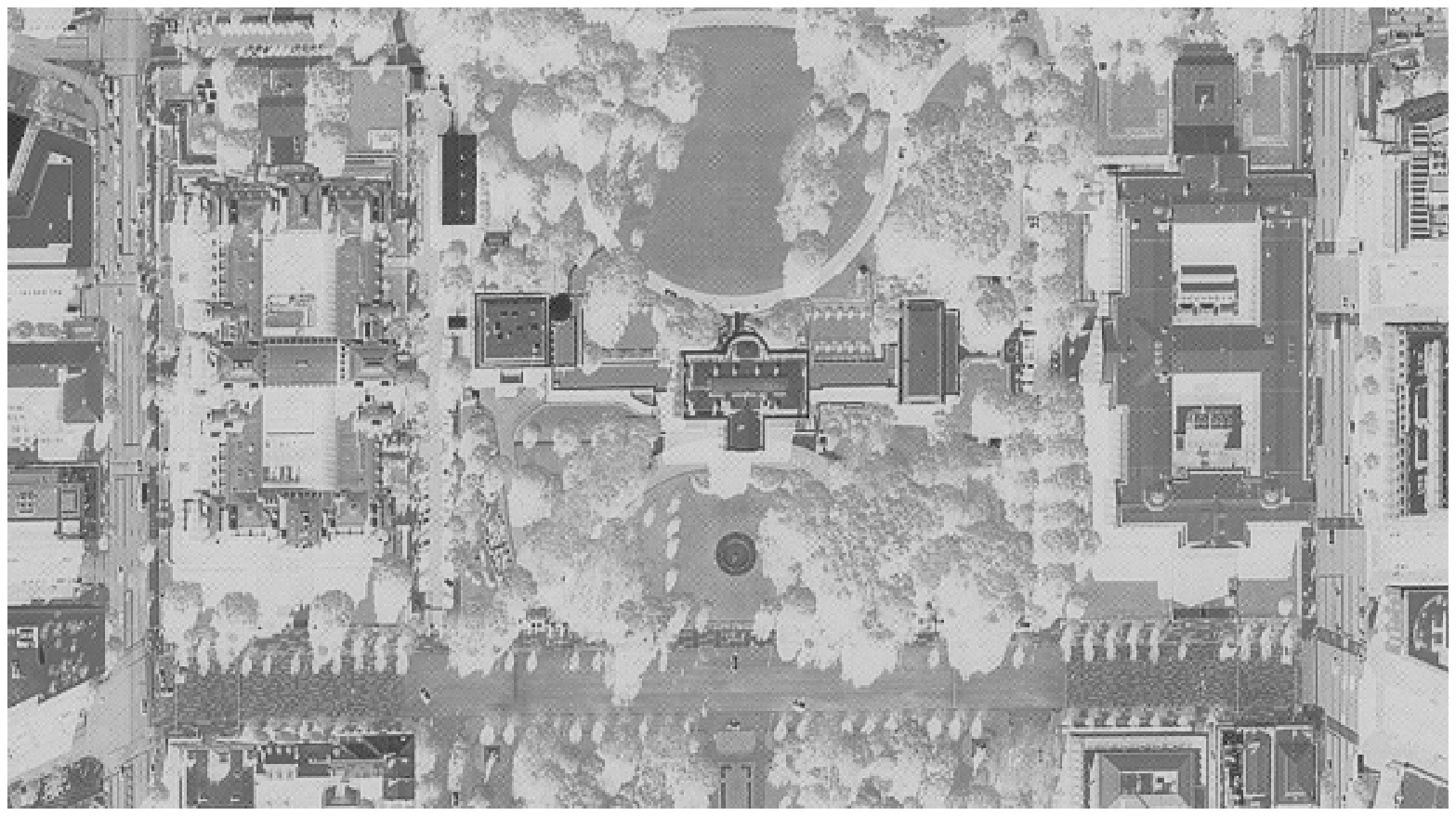}} \\
  \subfloat[Observations]{\includegraphics*[width=2.5in,height=1.5in,trim=80 130 80 90,clip]{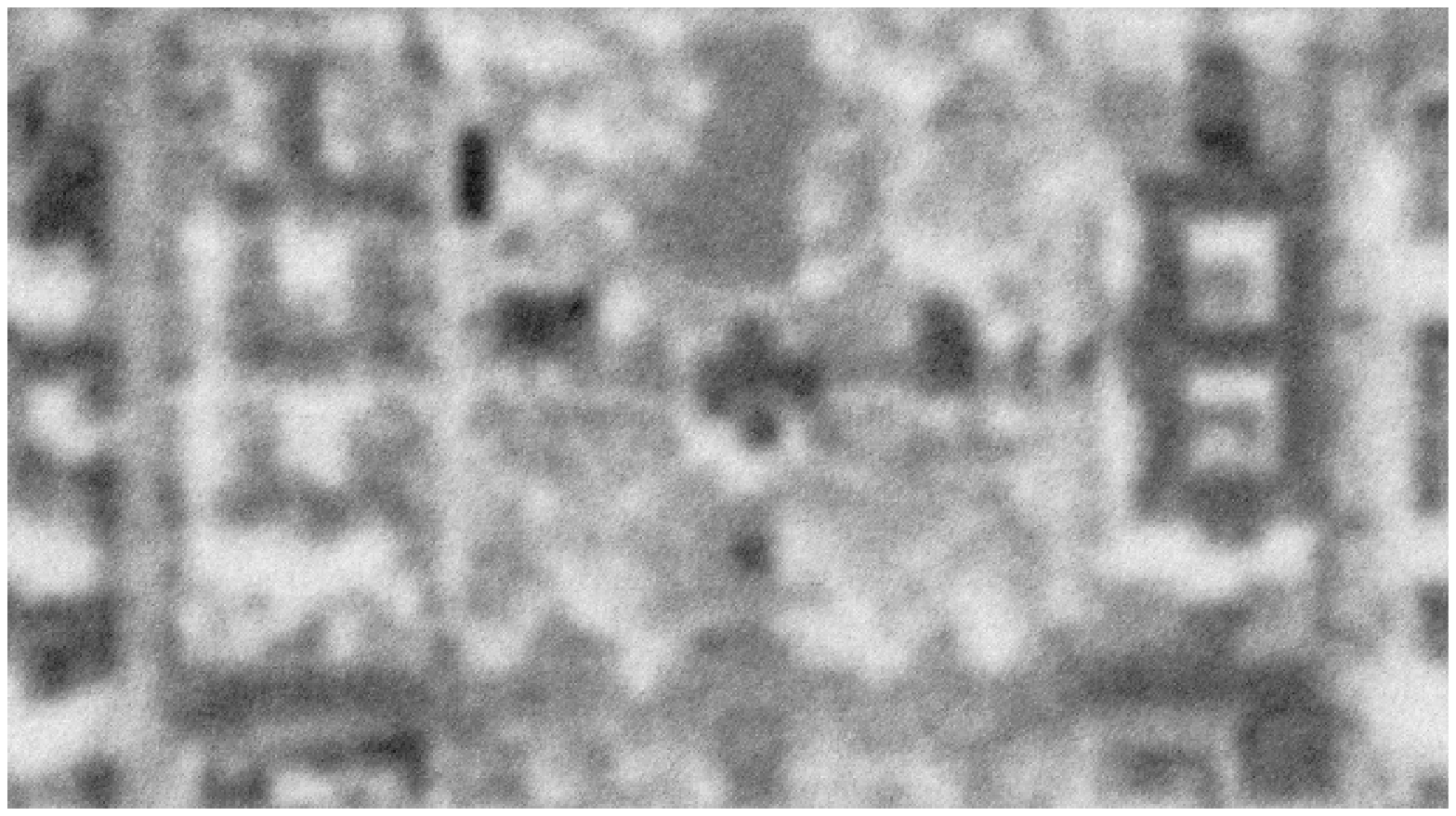} }
  \subfloat[Estimates]{\includegraphics*[width=2.5in,height=1.5in,trim=80 130 80 90,clip]{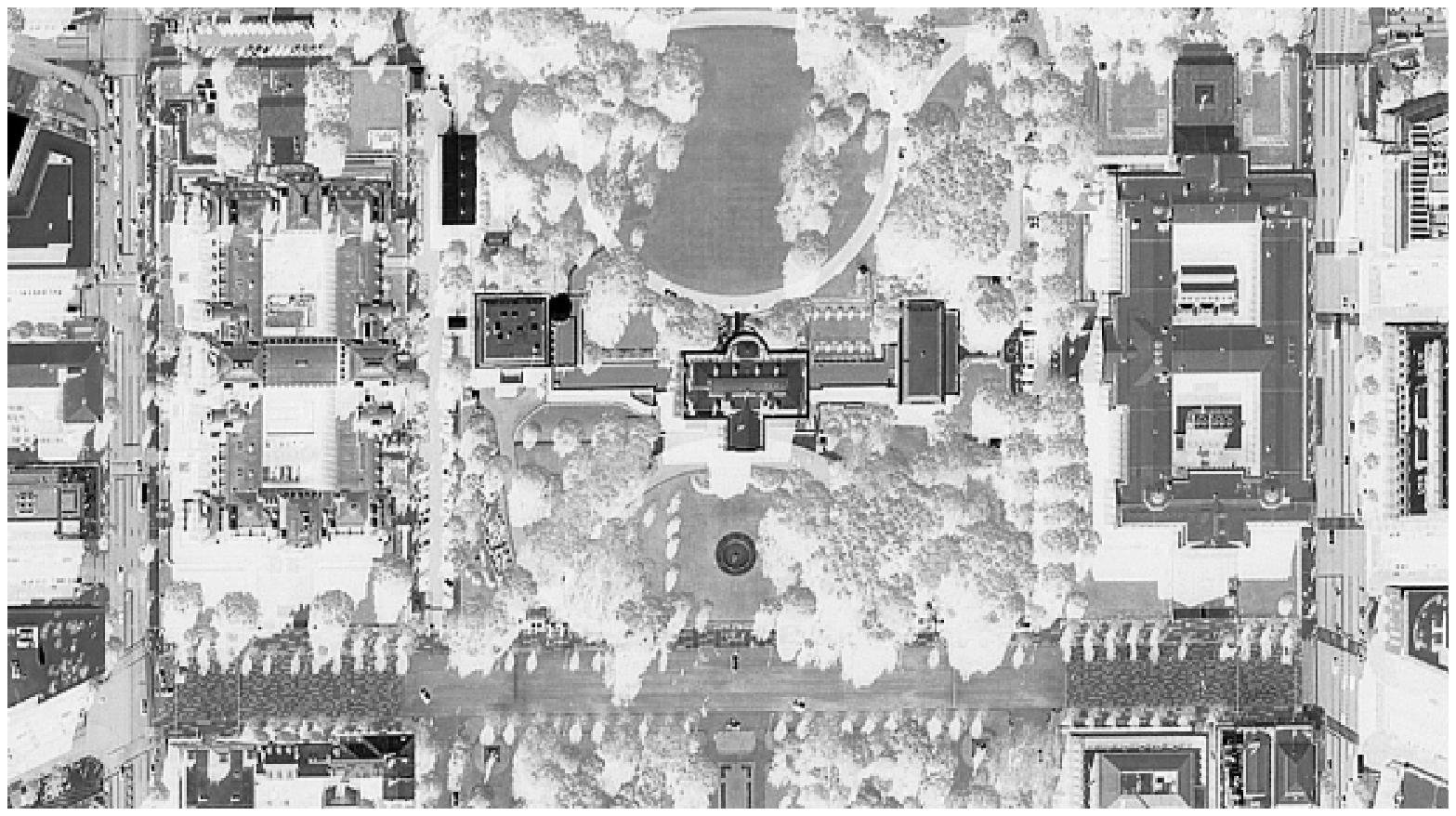} }
\caption{Example of images of the White House from a satellite and associated recovery of
the unknown signal $\theta$ using our proposed estimator.  In the left column (Observations)
the different amounts of blurring are due to varying atmospheric conditions and 
correspond to the forward operators $K_i$ in equation \eqref{eq:DWNP}. In the right column (Estimates)
we report the output of our estimator using the data in the left column.  Each row corresponds to making 
another observation $Y_i$ and updating our estimator with this new data. We emphasize that there are no choices
to be made by the data analyst; all tuning parameters are chosen in an automatic, data-dependent way.}
\label{fig:satelliteExample}
\end{figure}

\end{example}

This paper is organized as follows.  In Section \ref{sec:methodology}
we give a careful overview of our method and provide justification for 
the assumptions made, with greater exposition 
occurring in Appendix \ref{sec:sameSpectrum}.  In Theorem
\ref{thm:mainResult1} and Theorem \ref{thm:mainResult2}, we give supporting
theory for our estimator that both shows uniform consistency over
the parameter space and an asymptotic oracle inequality.  These results
show both that our estimator will get the correct answer eventually, no matter
the signal in our parameter space, and that our estimator makes essentially as efficient use
of the data as if we knew the signal $\theta$.
In Section \ref{sec:supportingExamples}, we provide an
example of our framework in action on simulated data.

{\it Notation.}  For $A \in \mathbb{C}^{p \times p}$
and $a \in \mathbb{C}$, define $A^*$ to be the Hermitian adjoint of $A$.  Correspondingly, define 
$|a|^2 = a^*a$ and $|A|^2 = A^*A$ to be the squared complex modulus of a scalar and matrix, respectively.
Likewise, for any vector $x \in \mathbb{C}^p$, $||x||^2 = x^*x$.
If $A A^* = I_p = A^*A$, then we say that $A$ is unitary.  
We utilize bold faced font for vectors: $\mathbf{b}_n \in \mathbb{C}^p$ with its $j^{th}$ entry
notated $b_{nj}$ and the subscript $n$ indicates dependence on the sample size. Similarly,
$A_{nj}$ is the $j^{th}$ element of the main diagonal of the matrix $A_n$.
We abuse notation slightly by using $\lamVec$ as both a vector in $\mathbb{C}^p$ and as a
function from $\mathbb{C}^p$ to $\mathbb{C}^p$ given by component-wise multiplication.

\section{Methodology and main results}
\label{sec:methodology}
We begin this section by making
the following assumptions, building on the notation introduced in
equation (\ref{eq:DWNP}):

\begin{itemize}
\item[(A1)] The noise parameter $\ep > 0$ is known.
\item[(A2)] The $(K_i)_{i=1}^n$ are known smoothing matrices.
\item[(A3)] There exists a unitary matrix
            $\Psi \in \mathbb{C}^{p \times p}$ 
            and diagonal matrices $D_i$
            such that $K_i = \Psi D_i \Psi^*$ for all $i = 1, \ldots, n,\ldots$.
\item[(A4)] There exists an $N<\infty$ such that for all $j$ there exists an $1 \leq i_* \leq N$ 
            such that $|D_{i_*j}| > 0$.
\item[(A5)] Define $\Delta_n := \sum_{i=1}^n |D_{ij}|^2$. Then the $(D_i)$ are such that
            \begin{equation*}
            \lim_{n\rightarrow \infty} \frac{\max_j \Delta_{nj}}{\min_j \Delta_{nj}} < \infty.
            \end{equation*}
\end{itemize}

Assumptions (A1) and (A2) are very standard in the statistical inverse problem literature.
We discuss a strategy for estimating $\ep$ in Section \ref{sec:estimatingVariance}.
Assumption (A4) is also commonly made and it ensures that, at some point,
the entire signal $\theta$ is identified and loosely corresponds to the intersection
of the null spaces of the $(K_i)_{i=1}^n$ eventually only containing the zero vector.
Assumption (A5) merely prevents a pathological case where the $K_i$ are becoming more
ill-conditioned without bound as $n\rightarrow\infty$.  
Assumption (A3) is crucial to our method
and while the reason for it will become clear, the following theorem
provides a general family of matrices that satisfy it:

\begin{theorem}
\label{thm:sameEvecsMain}
If the $(K_i)_{i=1}^n$ all correspond to the convolution operation,
then there exists a unitary matrix $\Psi$ and
a sequence of diagonal matrices $(D_i)_{i=1}^n$, all of which could have complex entries,
such that (A3) holds. If
$\theta$ is a one (two)-dimensional signal, then the $K_i$ are (block) circulant and
the entries of the matrix
$\Psi$ are the discrete one (two)-dimensional Fourier basis and
the entries of $D_i$ are the corresponding
discrete one (two)-dimensional Fourier coefficients.
\end{theorem}

Hence, we see that (A3) is more general than the convolutional 
assumption made in \cite{piana1996} and many other works concerning statistical inverse
problems.  
See Appendix \ref{sec:sameSpectrum} for a proof of Theorem \ref{thm:sameEvecsMain} 
and an investigation
into more general families of matrices that satisfy assumption (A3).

\subsection{Overview and main results}
An overview of our procedure is as follows.
The parameter $\theta$ and each observation $\mathbf{Y}_i$ 
us rotated by $\Psi^*$.
The rotated $\mathbf{Y}_i$'s are combined together to form
a sufficient statistic $\mathbf{B}_n$.
The estimators we consider
are of the form $\hat \theta = \Psi\lamVec(\mathbf{B}_n) := \Psi(\lambda_j B_{nj})_{j=1}^p$.
Define this set of estimators to be
\begin{equation}
\mathcal{E} = 
\{ \hat \theta = \Psi \lamVec (\mathbf{B}_n) : \lamVec \in \mathbb{C}^p\}.
\label{eq:mathcalE}
\end{equation}
We choose from the estimators in $\mathcal{E}$ using a combination of minimizing an empirical estimator of the risk
and some additional regularization parameters.  
Define our estimator to be $\hat \theta_n = \Psi \hat\lamVec(B_n)$, where
\begin{equation}
\hat \lambda_j =
\left(1 - \frac{\Omega_n^2\ep^2}{\Delta_{nj}|B_{nj}|^2} \right)_{+}.
\end{equation}
The form of this estimator is derived in the text containing and preceding 
equation \eqref{eq:mainEstimator}. We set 
$\Omega_n^2 := (p-2)\left(1 + \frac{\max_j \Delta_{nj}}{\min_j \Delta_{nj}}\right)$.
Note this choice of $\Omega_n^2$ is motivated by \citet{brown2011} in which it 
is shown that ensemble minimaxality in the heteroscedastic case holds for 
the soft thresholded James-Stein type estimator.

We  define our loss function to be the $l^2$ norm 
with associated risk
\begin{equation}
R(\hat \theta,\theta) := \mathbb{E}||\hat \theta - \theta ||^2
\label{eq:thetaRisk}
\end{equation}
and set 
$\Theta := \{\theta: ||\theta||_2^2 \leq \parm\}$ 
for any $0 < \parm < \infty$. Then
\begin{theorem} 
Under assumptions (A1) - (A5),
\begin{equation}
\limsup_{n \rightarrow \infty} \, \sup_{\theta \in \Theta} \, \gamma_n^{-1}
R\left(\hat \theta_n,\theta \right) < C < \infty
\end{equation} 
where 
\[
\gamma_n = \max_j \frac{\ep^2}{\Delta_{nj}}.
\]
\label{thm:mainResult1}
\end{theorem}

If $D_{ij} \equiv D_j$ for some $D_j \in \mathbb{C}$, then
$\gamma_n \asymp 1/n$; that is the parametric rate.  However, the forward operators $(K_i)$
in effect ensure that each observation doesn't decrease the risk equally.  The quantity $\Delta_{nj}$
relates to how much information is present in the first $n$ observations about the $j^{th}$ component
of $\Psi^*\theta$.  

Additionally, we can compare our estimator to the $\mathcal{E}$-oracle $\theta_*$

\begin{theorem}
Suppose assumptions (A1) - (A5) and let
\[
R(\theta_*,\theta) := \min_{\hat \theta \in \mathcal{E}} R(\hat \theta,\theta)
\] 
be the risk of the $\mathcal{E}$-oracle.  Then
\begin{equation}
R\left(\hat \theta_n,\theta \right) \leq R(\theta_*,\theta)(1 + O(1)),
\end{equation} 
where the term $O(1)$ does not depend on $\theta$.
\label{thm:mainResult2}
\end{theorem}

An interesting extension of this model is to the random operator setting.  That is,
what is the impact of having $K_i$ being drawn from some distribution?  We 
answer this question in an interesting case.

\subsection{Random Eigenvalues}
Suppose that the $(K_i)$ are random operators such that 
$K_i = \Psi D_i \Psi^*$ for all $i= 1,2,\ldots$
and $\text{diag}(D_i) \stackrel{i.i.d}{\sim} \mathcal{D}$, where $\mathcal{D}$ 
is any $p$-variate complex distribution
that doesn't have too much mass near zero.
Specifically,


\begin{itemize}
\item[(B4)] The distribution $\mathcal{D}$ is such that 
            there exists an $a$ where for $0 \leq \tau \leq a$
            \[
            \mathbb{P}_{\mathcal{D}}\left(|D_{1j}|^2 < \tau\right) = (\tau)^\rho.
            \]
\end{itemize}

This is a stochastic extension of assumption (A4) as it allows the random 
eigenvalues to be arbitrarily close to zero in magnitude but
with the probability of them being small going to zero at an appropriate rate.
Lastly, let $(W_i)$ and $(D_i)$ be mutually independent.

\begin{theorem}
Suppose assumption (B4) holds with some $\rho > 1$. Then
\begin{equation}
\lim_{n\rightarrow\infty} \sup_{\theta \in \Theta} 
\mathbb{E}_{(D_i),(Y_{ij})} \left|\left|\hat \theta_n - \theta\right|\right|^2 = 0
\end{equation}
where $\mathbb{E}_{(D_i),(Y_{ij})}$ corresponds to integration with respect
to the joint distribution of $(D_i)$ and $(Y_{ij})$.
\label{thm:randomEvalsUniformlyConsistent}
\end{theorem}
\subsection{Rotations, estimators, and tuning parameter selection}
Returning to equation (\ref{eq:DWNP}), for $i = 1, 2, \ldots$ we 
define $\mathbf{X}_i := \Psi^*\mathbf{Y}_i$, $\be := \Psi^*\theta$, 
and $\mathbf{Z}_i := \Psi^*\mathbf{W}_i$.
Then it follows that
\begin{equation}
\mathbf{X}_i = D_i \be + \ep \mathbf{Z}_i.
\label{eq:DWNPbeta}
\end{equation}
Note that in this case $\mathbf{Z}_i \stackrel{i.i.d}{\sim}
CN(0,I_p,\Psi\Psi^{\top})$\footnote{A complex normal has an extra parameter compared
with a real normal.  For a zero mean
complex normal random variable $\mathbf{Z}$, this is denoted
$CN(0,\mathbb{E}ZZ^{*},\mathbb{E}ZZ^{\top}$).}.
It is also convenient to look at equation \eqref{eq:DWNPbeta} component-wise,
\begin{equation}
X_{ij} = D_{ij} \be_j + \ep Z_{ij}
\label{eq:DWNP2comp}
\end{equation}
for $j = 1, \ldots, p$. Note that for these multiplications
to be defined, we have to think about $\mathbb{R}^p$ being embedded in
$\mathbb{C}^p$ by having imaginary part equal to zero.  We follow
this convention without comment in what follows.

\begin{remark}
Note that the $(\mathbf{Z}_i)$ are degenerate complex Gaussian vectors
in the following sense:
if we think of a $p$ dimensional complex Gaussian as a $2p$ dimensional real Gaussian with
some covariance matrix, then the Gaussian actually has values in a $p$ dimensional 
subspace of $\mathbb{R}^{2p}$. Thus the random variables
don't have a density with respect to Lebesgue measure on the full space $\mathbb{C}^p$,
among other complications.
\end{remark}
\begin{remark}
Commonly, the sequence space formulations found in equation \eqref{eq:DWNPbeta} and
equation \eqref{eq:DWNP2comp} are accomplished by a real, orthogonal matrix instead of a
complex, unitary one.  Allowing for the sequence $(K_i)_{i=1}^n$
to share the same eigenvectors necessitates permitting $\Psi$ to be complex.
This makes equation \eqref{eq:DWNP2comp} more complicated  
than the conventional normal means problem in at least two ways.
First, as stated above, the random variables are complex.  Second, and more
importantly, the model is heteroscedastic.  This leads to a much more involved theory 
than in the homoscedastic case, such as in \citet{brown1975}, 
and is still the topic of contemporary research
\citep{brown2011}.
\end{remark}

Lastly, define
\begin{equation}
B_{nj} := \frac{\sum_{i=1}^n D_{ij}^*X_{ij}}{\sum_{i=1}^n |D_{ij}|^2} = 
\be_j + \ep\Delta_{nj}^{-1/2}Z_j
\label{eq:Bnj}
\end{equation}
where $\Delta_{nj} := \sum_{i=1}^n |D_{ij}|^2$.

This quantity is particularly important, as evidenced by the following theorem
\begin{theorem}
Under the model introduced in equation \eqref{eq:DWNP} and (A1) - (A4), 
the random vector $\mathbf{B}_n := (B_{nj})_{j=1}^p$
is sufficient for $\be$ in equation (\ref{eq:DWNPbeta}).
\label{thm:sufficient}
\end{theorem} 
This claim can be seen by noting that the map $\Phi^*$ is measure preserving.

As $\Psi$ is also unitary, we can define an equivalent risk
to the one defined in equation \eqref{eq:thetaRisk} in terms of $\be$
\begin{equation}
R(\hat \theta,\theta) := \mathbb{E}||\hat \theta - \theta ||^2
= \mathbb{E}||\Psi^*(\hat \theta - \theta) ||^2 =
\mathbb{E}||\hat \be - \be ||^2 
 =: R(\hat \be,\be).
\label{eq:risk}
\end{equation}

Any
risk computations made under the data, which is $(X_i)_{i=1}^n$ in our
notation, are equivalent to those
made under a sufficient statistic \citep[Theorem 7.1]{bahadur1954}.  
By Theorem \ref{thm:sufficient}, $\mathbf{B}_n$ is
sufficient for $\beta$ and hence for all measurable functions of the data
that are not functions of $\mathbf{B}_n$, there exists an estimator with equal risk
that is a function of
$\mathbf{B}_n$.  In fact, the expectations in equation \eqref{eq:risk} are equivalent 
under $(X_i)_{i=1}^n$ and $\mathbf{B}_n$.
Therefore, for each $n$, we can treat $\mathbf{B}_n$ as the data
and formulate estimators based upon it.  

To develop an automatic procedure for signal estimation in 
sequential inverse problems we begin by regularizing an 
unbiased estimator of $\be$ through the use of  
a smoothing parameter vector.  We choose this smoothing parameter
by minimizing an estimate of the risk.  
This type of procedure, known generally as unbiased risk estimation,
has been revisited regularly in many fields for solving various
problems related to denoising \citep{stein1981,donoho1995}.
However, as inverse problems generally result in unstable estimators of both
the parameter $\be$ and the risk $R$,
we compensate by including additional regularization.

The specifics of our approach are related to the procedure 
found in \citet{beran2000}.  
However, the goal in \citet{beran2000}, unlike our paper, is the estimation of the 
regression function in an assumed linear model instead
of the coefficients themselves. 
That is, referring to the notation in equation \eqref{eq:DWNP},
the estimation of $K_i\theta$ instead of the estimation of $\theta$.
This is an important distinction as both estimating $\theta$ is intrinsically harder
than estimating $K_i\theta$ and $\theta$ is the object of actual interest.  The practical implications
of these differences is that only minimizing an unbiased estimate of risk,
as is the procedure in \citet{beran2000}, provides
insufficient regularization. As well, the theoretical justification that appears in \citet{beran2000},
is essentially entirely asymptotic in $p$.
This is a regime we do not consider
relevant for the problem at hand.

To begin to formulate an estimator of $\be$, and therefore $\theta$, we state the
following:
\begin{proposition}
Define $\hat \psi_j := (|B_{nj}|^2 - \ep^2/\Delta_{nj})/|B_{nj}|^2$.
Then the random function
\begin{equation}
\hat R_n(\lamVec) := \sum_{j=1}^p (\lambda_j - \hat\psi_j)^2|B_{nj}|^2
\label{eq:riskEstimateFinal}
\end{equation}
provides, up to a constant independent of $\lamVec$, an unbiased estimate of $R(\lamVec)$.
Additionally,
\begin{equation}
\min_{\lambda \in \mathbb{C}^p} R(\lamVec) = \min_{\lambda \in \hypSq} R(\lamVec) 
\end{equation}
where $\hypSq = [0,1]^p$ is the $p$ dimensional hypersquare.
\label{prop:hatRn}
\end{proposition}

The first part of the proposition provides an unbiased estimate of the risk
while the second
part implies that we gain no improvement in risk by 
allowing $\lamVec$ to have values outside of $\hypSq$.

Using $\hat R_n$ from \eqref{eq:riskEstimateFinal}, 
define for any $\gen \subseteq \hypSq$ 
\begin{equation}
\hat{\lamVec}^{\gen} := \argmin_{\gen} \hat{R}_n(\lamVec)
\end{equation}
which produces an estimator of $\be$ via
\begin{equation}
\hat{\be}^{\gen} := \hat{\lamVec}^{\gen}(\mathbf{B}_n).
\label{eq:filterEstimator}
\end{equation}
Lastly, we recover an estimate of $\theta$ by
forming $\hat{\theta}^{\gen} := \Psi \hat{\be}^{\gen}$.

As any choice of $\gen$ results in an estimator 
$\hat{\be}^{\gen}$ via the above machinery, there are in principle many possible
choices.  We focus on $\gen = \hypSq$, which by inspection of equation \eqref{eq:riskEstimateFinal},
results in 
\begin{equation}
\hat{\lamVec}^{\hypSq} = \left(1 - \frac{\ep^2}{\Delta_{nj}|B_{nj}|^2} \right)_{+}
\end{equation}
where as usual $(\cdot)_+ = \max(\cdot,0)$ is the soft thresholding function. 
Other choices can and should be explored in further research
into estimation in sequential inverse problems such as
$\mono := \{ \lamVec \in \hypSq: \lambda_1 \geq \lambda_2 \geq \ldots \geq \lambda_p\}$,
which induces a monotonicity constraint on the estimated coefficients,
or block methods of piecewise constant
weights \citep{cavalier2002a}. 

Additionally, the aforementioned Tikonov-Phillips regularization
and Landwieber iterations methods correspond to 
specific subsets of $\hypSq$. 
The Tikonov-Phillips estimator takes the form
\begin{equation}
\hat{\be_j}^{\gamma} := \sum_{i=1}^n \frac{D_{ij}^* X_{ij}}{|D_{ij}|^2 + \gamma}
\label{eq:TPestimator}
\end{equation}
which can be rewritten as an element of $\mathcal{E}$ by defining
\begin{equation}
\lambda_j^{\gamma} := \frac{\Delta_{nj}}{\Delta_{nj} + \gamma}
\end{equation}
with associated estimator 
$\hat{\be}^{\gamma} = \lamVec^{\gamma}(\mathbf{B_n})$.

The Landwieber iterations estimator is by nature iterative.
However, it has an equivalent formation in the form of the
following linear smoother
\begin{equation}
\lambda_j^{(\gamma,\tau)} = (1- [1-\tau\Delta_{nj}]^{\gamma})
\label{eq:LIestimator}
\end{equation}
where $\gamma$ corresponds to the number of iterations and $\tau$ is a relaxation parameter.
The associated estimator is $\hat{\be}^{(\gamma,\tau)} = \lamVec^{(\gamma,\tau)}(\mathbf{B}_n)$.
Hence, this procedure generalizes the results 
in \citet{piana1996}
by providing a principled tuning parameter selection method.

A problem arises if we choose smoothing parameters in this fashion in inverse problems:
insufficient regularization.  This is due to $\hat R_n$ being
an unstable estimate of $R$ for the same reason as $\mathbf{B}_n$ is an unstable estimator of $\be$.

Instead of regularizing the risk estimator, we modify the weights directly
to provide additional
regularization.  However, we record our belief that regularizing $\hat R_n$ by limiting 
how ill-conditioned the risk estimator can become and then minimizing
this biased estimator of the risk should provide a suite of interesting
estimators via the above machinery. Define
\begin{equation}
\hat{\lamVec} = \left(1 - \frac{\Omega_n^2\ep^2}{\Delta_{nj}|B_{nj}|^2} \right)_{+}
\label{eq:mainEstimator}
\end{equation}
where the parameter $\Omega_n^2$ is specified before Theorem \ref{thm:mainResult1}.  
Lastly, define our estimator of $\theta$
to be 
\begin{equation}
\hat \theta_n := \Psi \hat{\lamVec}(\mathbf{B}_n).
\end{equation}

\section{Computational concerns, variance estimation, and alternate methods}
\subsection{Computations}
The specifics of the computation of an estimator $\hat{theta}^{\gen}$
depend on the subset $\gen$.  However, $\hat R_n$ is a convex objective
function.  Hence, if $\gen$ is a convex subset of $\mathbb{R}^p$, then the solution
can be found both efficiently and uniquely.
Of the estimators mentioned above, all except $\mono$ have a closed form solution
and therefore trivial computation.  The minimization of $\hat R_n$ over $\mono$
can be accomplished by a well known algorithm called
Pooled Adjacent  Violators (PAV) \citep{robertson1988} that transforms
the least squares solution $\hat\psi$ into the monotone solution 
by taking weighted averages of adjacent elements of $\hat\psi$ 
that violate the monotonicity constraint.

Additionally, in the case of convolution, the vector $\Delta_n$ and the
random variables $(X_i)$ can be computed via the Fast Fourier Transform, which 
implies $O(p \log p)$ computations and is of course the archetypal instance of an efficient
algorithm.  However, for more general matrices $K_i$, the eigenvectors must be computed
via a conventional eigenvector solver, which necessarily has computational complexity
$O(p^3)$.  This could become prohibitive for large scale problems.  There do exist modern
approximation methods for eigenvalues and eigenvectors that could be used instead, such 
as in \citet{halko2009}.  However, we do not explore this idea further in this paper.

An additional feature is that for the computation at step $n$, it is not necessary to
keep the entire history $(Y_i)_{i=1}^n$ and $(K_i)_{i=1}^n$.  Both
$\mathbf{B_n}$ and $\Delta_n$ can be computed from aggregate information.  Hence, we
can produce an estimate of $\theta$ given only access to 
a few summary statistics which get updated after each new observation.

\subsection{Estimating the variance parameter}
\label{sec:estimatingVariance}
Estimating the variance parameter can be accomplished in a consistent way by
setting aside a subsequence $\mathcal{N}$ of $\mathbb{N}$ and
computing the estimator
\begin{equation}
\hat \ep_{\text{con}}^2 := \frac{1}{pn'} \sum_{i \in \mathcal{N}} \sum_{j=1}^p 
\left(Y_{ij}^2 - \overline{Y}_j^2\right).
\end{equation}
Here, we have computed the estimator after the first $n'$ entries in $\mathcal{N}$.

Alternatively, we can take advantage of the observational process to acquire
a good estimate of $\epsilon$.  As we make observations, occasionally some will be
of exceptionally poor quality.  This observation will be less helpful for recovery
in general and provide almost no information about the higher order components of the vector $\be$.

Suppose now that $\mathcal{N}$ is the set of all indices $i$ such that $Y_i$ 
is a low quality observation; that is there exists a $p'$ 
such that for $j = p', \ldots, p$, the $|D_{ij}|^2$ are small. In general, $p'$ could depend on $i$, but
we do not consider this complication here. Form the following statistic
\begin{equation}
  \hat \ep_i^2 := \frac{1}{p - p'}\sum_{q = p'}^p |X_{iq}|^2.
\end{equation}

Then $\mathbb{E}\hat \ep_i^2 = \ep^2 + \frac{1}{p - p'}\sum_{q = p'}^p |D_{iq}|^2|\be_j|^2$
and we report $1/n' \sum_{i \in \mathcal{N}} \hat \ep_i^2$ as our estimator of $\epsilon^2$.
This is in general a biased estimator of the variance.  Nevertheless, it is still
useful.  First, it is conservative owing to its positive bias. 
Perhaps more importantly, this estimator provides an interesting situation where
the lowest quality parts of the lowest quality observations provide the best performance.

\subsection{Averaging is not enough}
\label{sec:avNotEnough}
In equation \eqref{eq:DWNP}, conventional statistical practice would
suggest averaging the observations $(Y_i)$ directly.  However,
we show here that this leads to suboptimal results.  Specifically, averaging
gives thes following model
\begin{equation}
\avgY = \avgK \theta + \frac{\ep}{\sqrt{n}} W
\label{eq:averageModel}
\end{equation} 
where, under assumption (A3), $\avgK := 1/n \sum_{i=1}^n K_i = \Psi \avgD \Psi^*$, \\
$\avgD := 1/n \sum_{i=1}^n D_i$,
$\avgY := 1/n \sum_{i=1}^n Y_i$, and $W \sim N(0,I_p)$.
This can also be equivalently expressed as
\begin{equation}
\avgB = |\avgD|^{-2} \avgD^*\avgX =   \beta + \frac{\ep}{\sqrt{n}}|\avgD|^{-2} \avgD^* \Psi^* W.
\label{eq:averageModel1}
\end{equation} 
Here, $\avgX = \Psi^* \avgY$.  We define the corresponding set of linear estimators to be 
$\avgE := \{ \hat \theta = \Psi \lamVec (\avgB) : \lamVec \in \mathbb{C}^p\}$. 

Note that we can write equation \eqref{eq:averageModel} without any assumptions about
the eigenvectors of the forward operators.  However, 
under assumption (A3), the following theorem supports forming estimators based 
on equation \eqref{eq:Bnj} instead of
equation \eqref{eq:averageModel}.
\begin{theorem}
Suppose for a fixed $\theta$, 
\[
R_1 = \inf_{\hat \theta \in \mathcal{E}} \mathbb{E}||\hat \theta - \theta||_2^2 \quad \text{ and } \quad 
R_2 = \inf_{\hat \theta \in \avgE} \mathbb{E}||\hat \theta - \theta||_2^2,
\]
where the expectations in $R_1$ and $R_2$ are under $\mathbf{B}_n$ and $\avgB$, respectively.
Then
\[
R_1 <^* R_2
\]
where `$<^*$' means `strictly less than except when $D_i \equiv D$ for all $i$ and some $D$.'
That is, the oracle linear risk based on equation \eqref{eq:Bnj} is strictly less than the oracle linear risk based
on equation \eqref{eq:averageModel}.
\label{thm:spectraThenAverage}
\end{theorem}
\begin{remark}
Note that the classic Tikonov-Phillips estimator based on the $\avgY$ is 
of the form
\[
\hthetaRidge = (\avgK^{\top}\avgK + \tau I)^{-1}\avgK^{\top} \avgY.
\]
This is equivalent to
\begin{equation}
\hthetaRidge = \Psi (|\avgD|^2 + \tau I)^{-1} |\avgD|^2 |\avgD|^{-2} \avgD^* \avgX = 
\Psi (|\avgD|^2 + \tau I)^{-1} |\avgD|^2 \avgB,
\label{eq:ridgeEstimator}
\end{equation}
and hence the Tikonov-Phillips estimator is in $\avgE$, among many others.
\end{remark}

An alternative approach relies on forming $\mathcal{K}_n := [K_1^{\top}, \ldots, K_n^{\top}]^{\top}$,
$\mathcal{Y}_n := [Y_1^{\top}, \ldots, Y_n^{\top}]^{\top}$, and $\mathcal{W}_n \sim N(0,I_{np})$.
Then it follows that
\begin{equation}
\mathcal{Y}_n = \mathcal{K}_n \theta + \ep \mathcal{W}_n.
\end{equation}
However, estimators based on this approach, such as spline type estimators, rely on accessing the entire
history of observations $(Y_i)$ and forward operators $(K_i)$.  This is computationally infeasible as this
means both keeping and repeatedly accessing the entire sequence of observations.  Hence, this approach doesn't
satisfy our requirement of an estimate at time $n$ being efficiently updatable to
a new estimate after recording $Y_{n+1}$.


\section{Supporting simulations}
\label{sec:supportingExamples}
\subsection{Description}
In this section we present visual results of using our estimator $\hat\theta_n$
to reconstruct various signals, 
given access only to smoothed and noisy, but repeated,  observations
of that signal.  
In both cases, we compare our estimator, $\hat \theta_n$ to $\hat \theta_{\text{ridge}}$ from 
equation \eqref{eq:ridgeEstimator}, with the smoothing parameter $\tau$ chosen by minimizing generalized
cross validation (GCV).  For a quantitative comparison, we use the normalized relative risk ($RR$)
given by
\begin{equation}
RR(\hat\theta,\theta) = \sqrt{\frac{R(\hat \theta,\theta)}{||\theta||^2}}.
\end{equation}
We estimate $RR$ by averaging 100 runs of our simulations.

For each of the signals introduced below, we set $p=256$
and fix the noise parameter $\epsilon$ to be such that
the signal-to-noise $:= ||\theta||_1/(p\epsilon) = 1$.
For these examples, we admit $K_i$  
that are an equally weighted mixture of three Gaussians,
normalized to have $l_1$ mass equal to 1,
with means $\mu_1 = -0.75$, $\mu_2 = 0.00$, and $\mu_3 = 0.50$, along with
standard deviations $\sigma_{iq} = 0.5 + E_{qi}$,
where $E_{qi} \stackrel{i.i.d.}{\sim} \text{exponential}(1)$ and $q = 1,2,3$.
Note that this implies that the $K_i$ are not symmetric.  Also, note that Gaussian-like smoothing represents one of the
worst cases as it exponentially down-weights the $\be_j$ for large $j$.

\begin{figure}
  \centering
  \subfloat{\includegraphics*[width=2.5in,trim=0 35 0 35,clip]{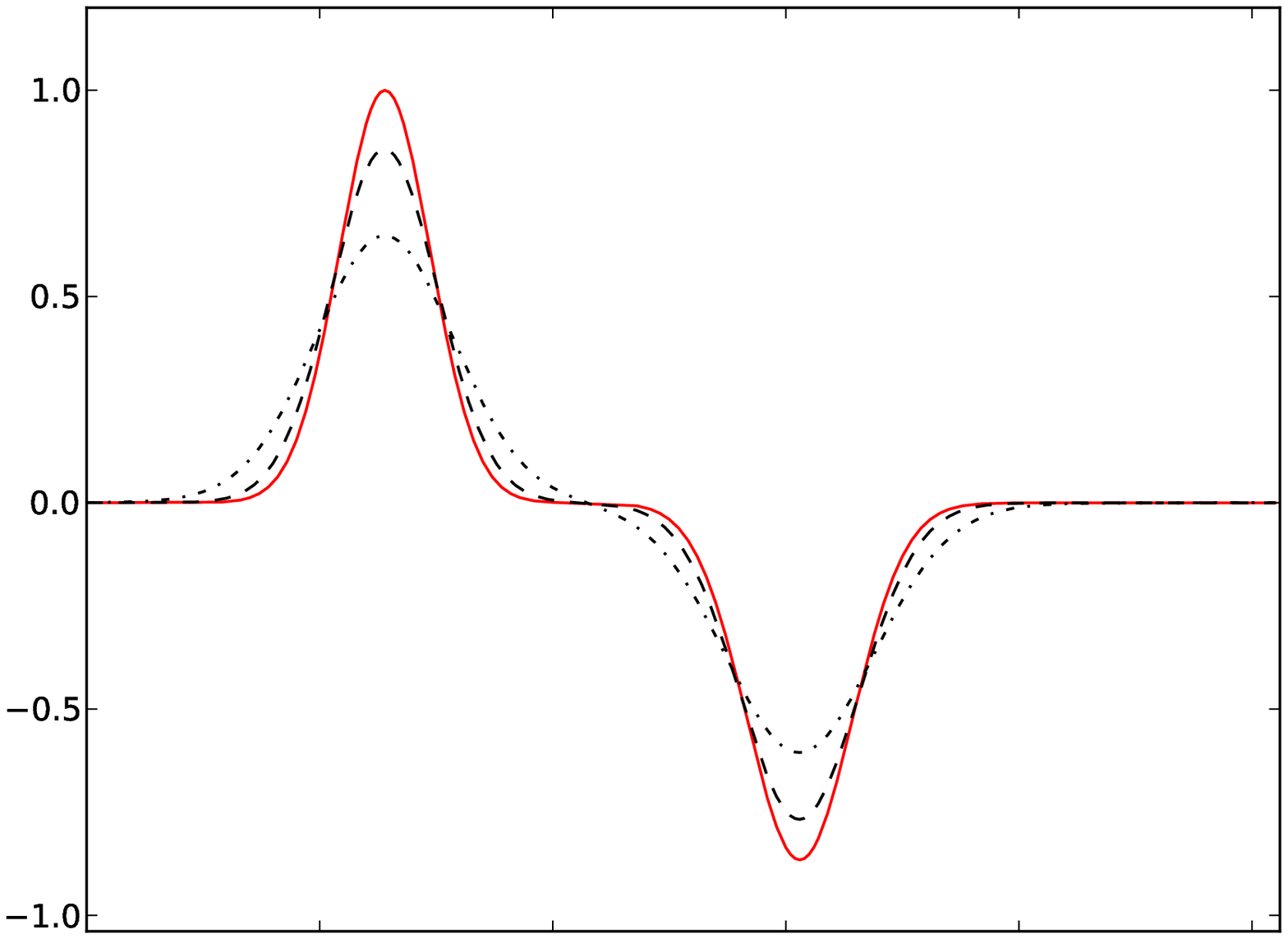}}
  \subfloat{\includegraphics*[width=2.5in,trim=0 35 0 35,clip]{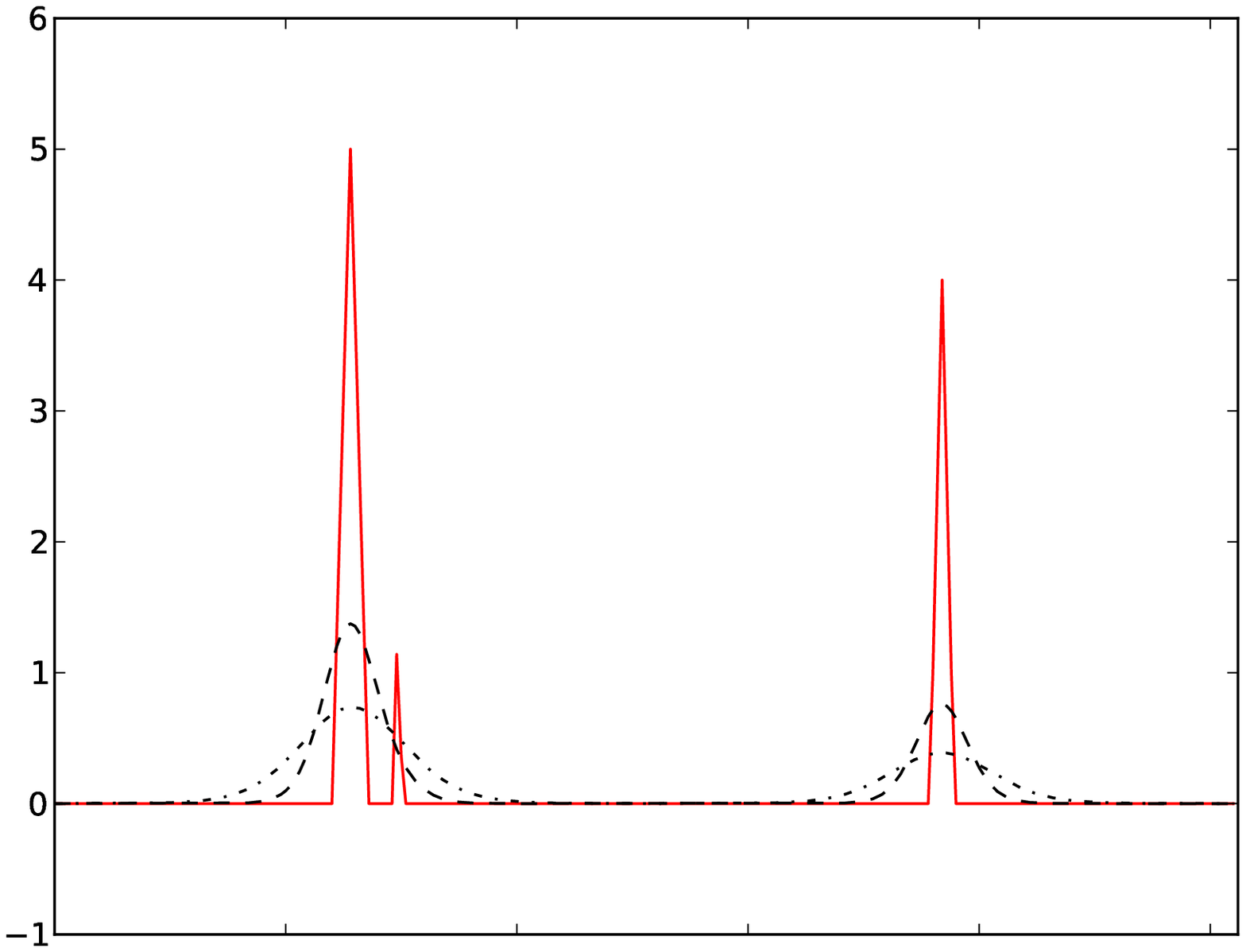}} \\       
  \subfloat{\includegraphics*[width=2.5in,trim=0 35 0 35,clip]{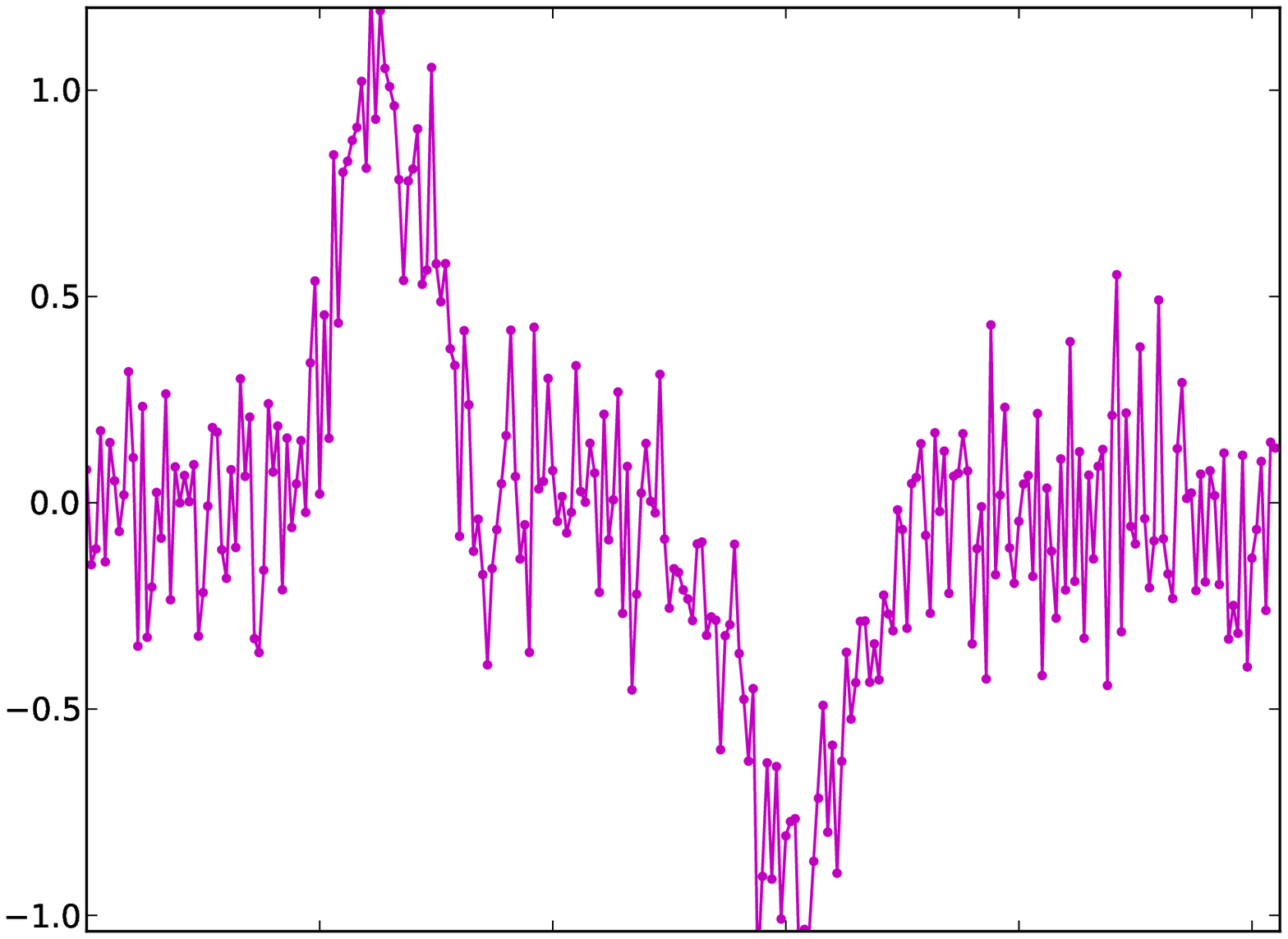}}
  \subfloat{\includegraphics*[width=2.5in,trim=0 35 0 35,clip]{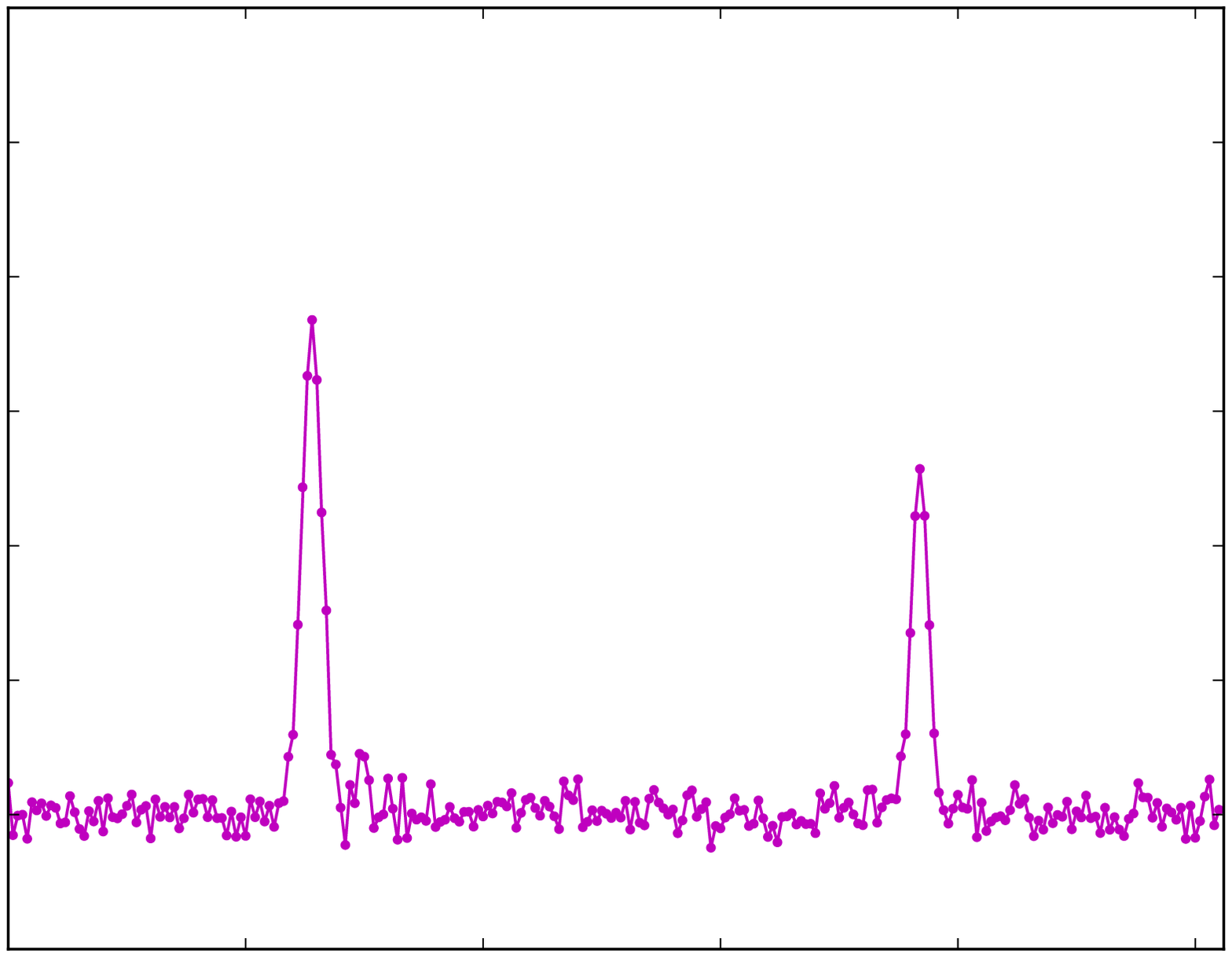}} \\
\caption{The left column corresponds to $\thetaS$ 
and the right column corresponds to $\thetaP$.  The top row
is a plot of the signal itself, along with the signal after the minimum (dashed line) and
maximum (dashed and dotted line) amount of smoothing.  The bottom row is an example of
the recorded data after corruption by smoothing and noise.   Notice that in $\thetaP$,
the smaller peak is completely obscured.}
\label{fig:1dsignals}
\end{figure}

We consider two signals for estimation, which we
refer to as $\thetaS$ and $\thetaP$ (Figure \ref{fig:1dsignals}).
The first signal, $\thetaS$,
is the sum of two Gaussians functions that are filtered by a 
Gaussian-tapered filter.  This filter is additionally  
enforced to be zero above the $p/2$ frequency.
Hence, $\thetaS$ is very smooth and compactly supported in the
frequency domain. This example is instructive as a smooth function 
should be well represented by the eigenvectors $\Psi$ of the smoothing operators $K_i$.  
Also, a compact representation in frequency domain will reveal the effectiveness of the
soft-thresholding in zeroing out the appropriate $B_{nj}$, ie: those that 
correspond to the $\be_j$ that are zero.
See the left column of Figure \ref{fig:1dsignals} for a plot of $\thetaS$
(top) along with a typical example of a noisy, smoothed version that comprises the 
recorded data (bottom).

Additionally, we consider the opposite situation by defining a signal
$\thetaP$ that is the sum of three sharp, non-smooth, peaks.  This signal
is difficult to represent with the eigenvectors of smoothing matrices but is
common in signal processing as it corresponds to both spectra from biochemical
analysis and nuclear magnetic resonance imaging (nMRI).  Note that the smallest
peak is completely obscured by the smoothing and noise.
See the right column of Figure \ref{fig:1dsignals} for a plot of $\thetaP$
(top) along with an example of a noisy, smoothed version (bottom).

\subsection{Results}
In estimating either signal, $\thetaS$ or $\thetaP$,
the estimator $\hat\theta_n$ converges rapidly to the truth.  See Table \ref{tab:RR}
for the $RR$ of $\hat\theta_n$ and $\hthetaRidge$ used on both signals.  
In each case, for $n = 50$, the $RR$ are approximately the
same, with $\hthetaRidge$ having a slight edge.  Every sample size
thereafter shows substantial advantage of $\hat\theta_n$ over $\hthetaRidge$,
culminating with a factor of two improvement in $RR$ after $n= 300$ observations.

For estimating $\thetaS$, both estimators have substantial oscillations
for low sample sizes.  However, due to $\hat\theta_n$ having
a soft-thresholding effect, some of the entries in our estimator of $\be$
are zeroed out. In contrast, $\hthetaRidge$ only shrinks
the coefficients and hence still has substantial fluctuations after $n=300$ observations.
See Figure \ref{fig:1dsmooth} for graphical results.

\begin{figure}
  \centering
  \subfloat{\includegraphics*[width=2.5in,trim=0 30 0 0,clip]{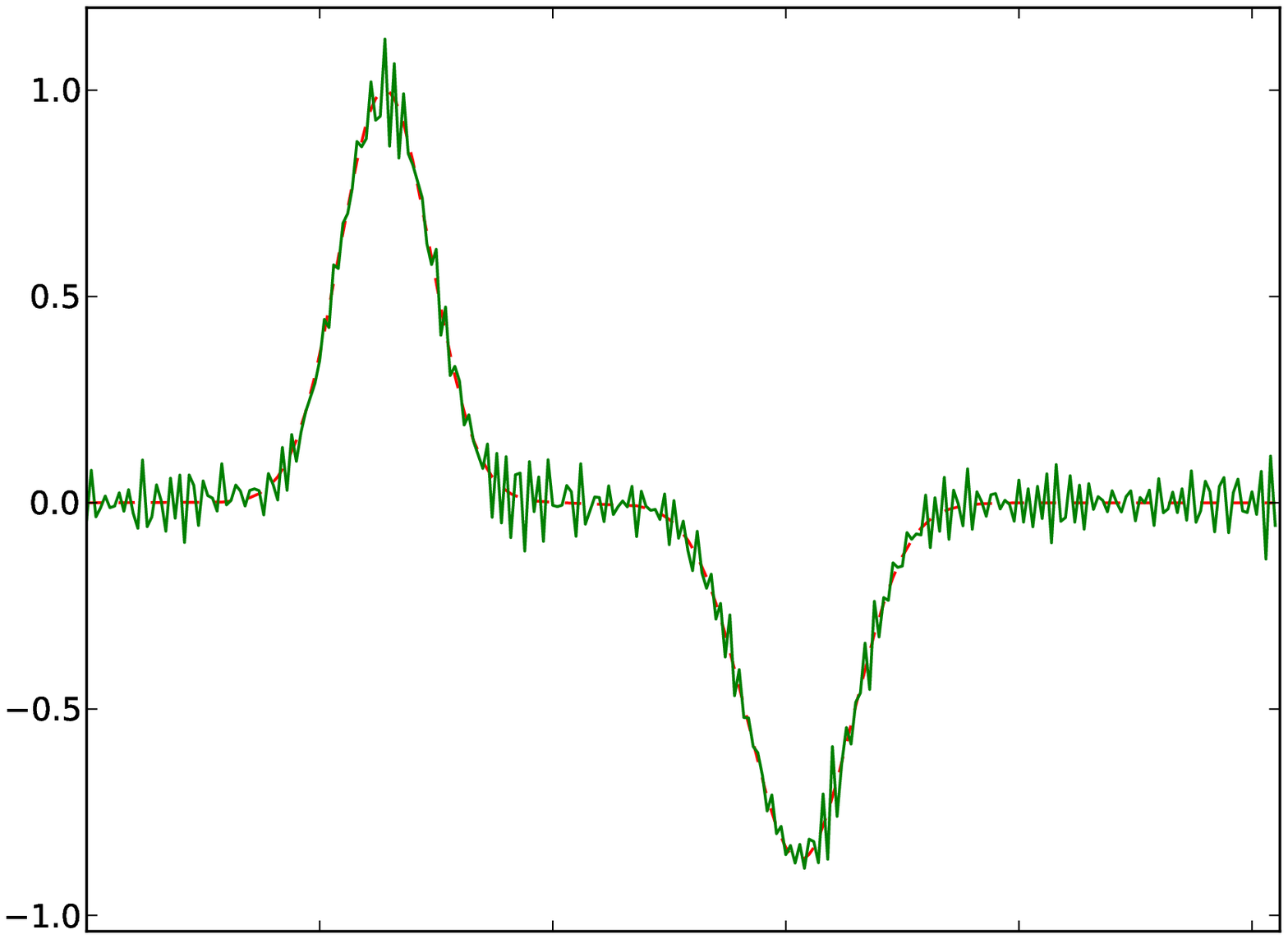}}
  \subfloat{\includegraphics*[width=2.5in,trim=0 30 0 0,clip]{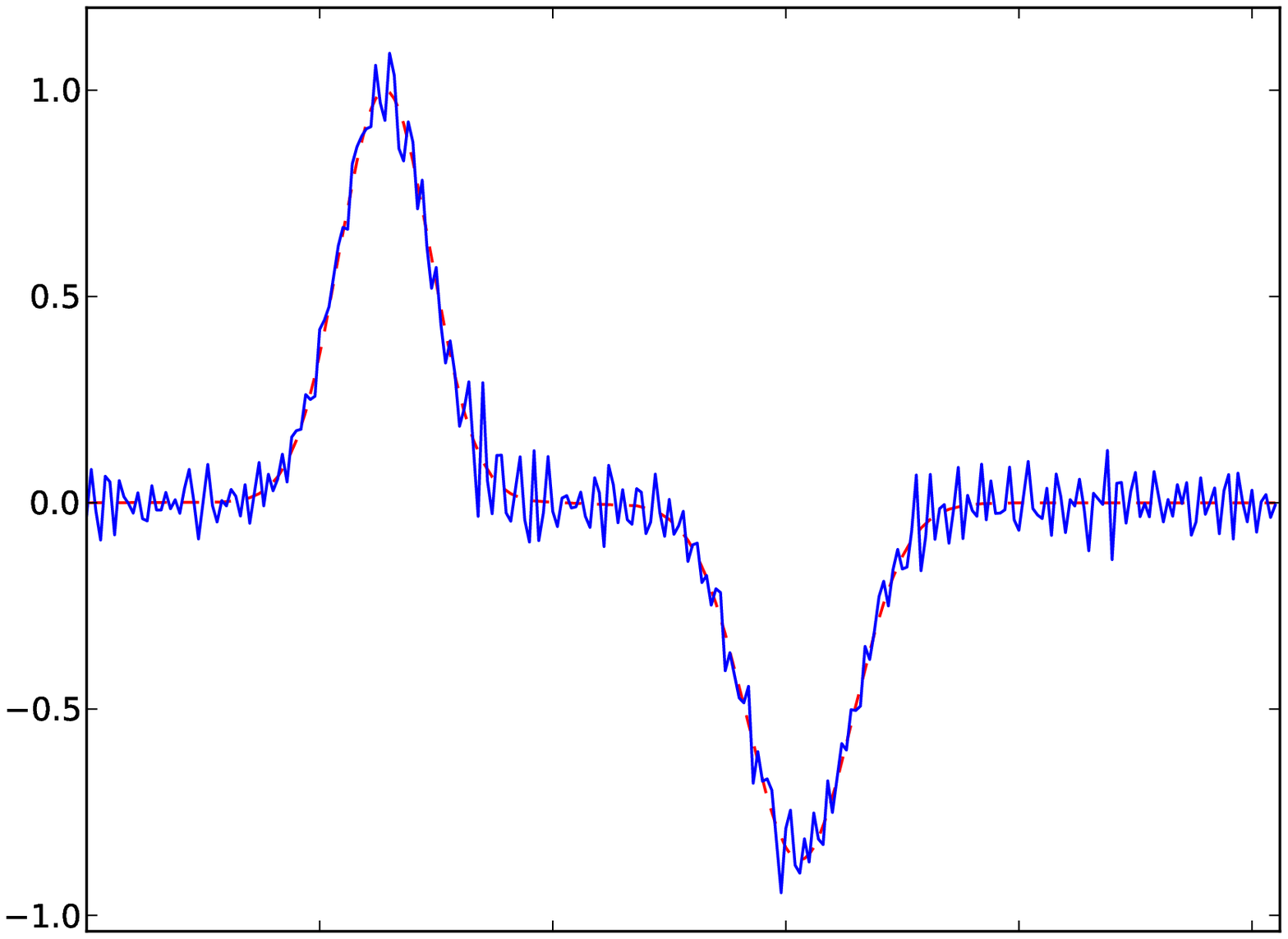}} \\       
  \subfloat{\includegraphics*[width=2.5in,trim=0 30 0 0,clip]{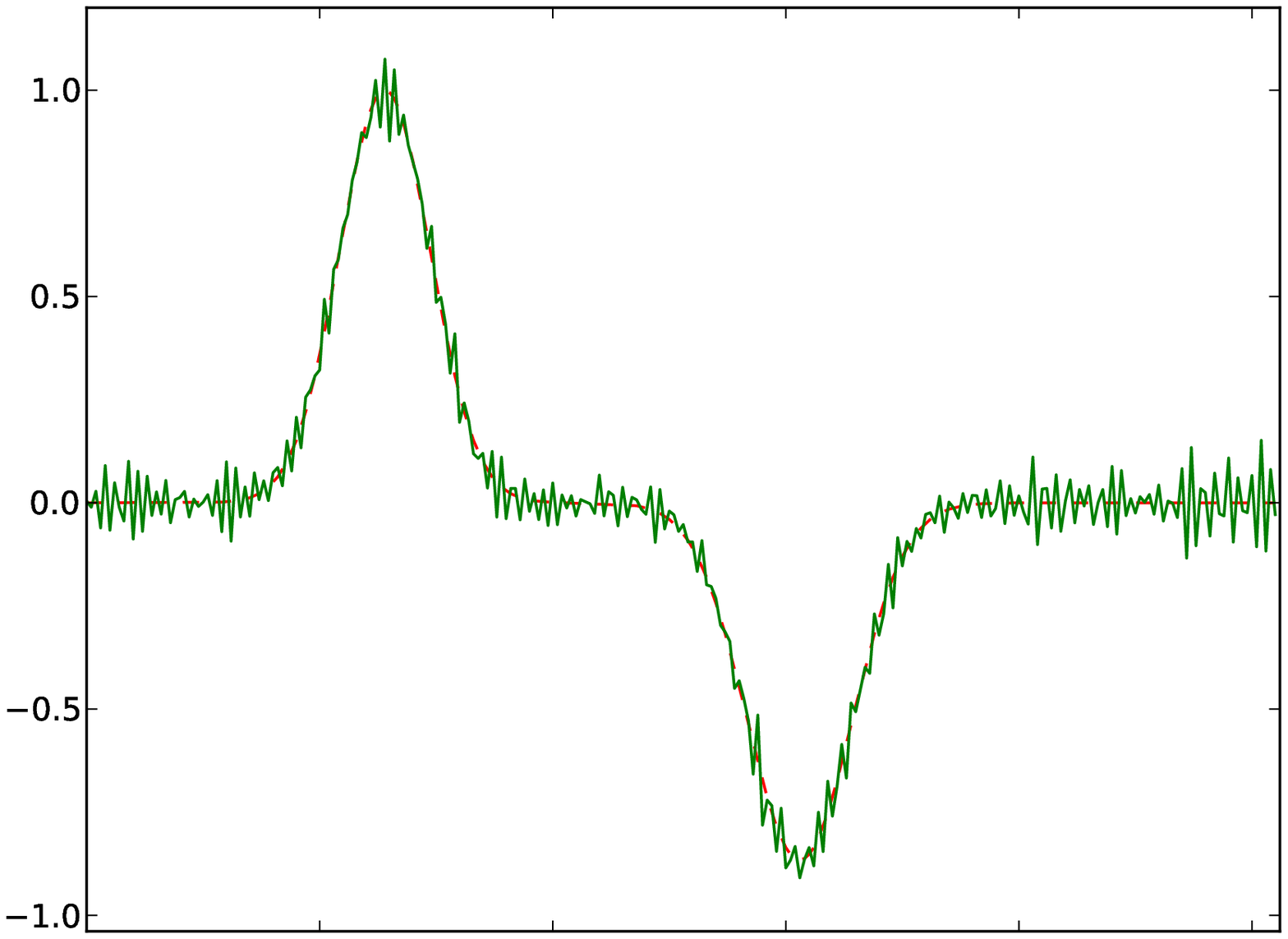}}
  \subfloat{\includegraphics*[width=2.5in,trim=0 30 0 0,clip]{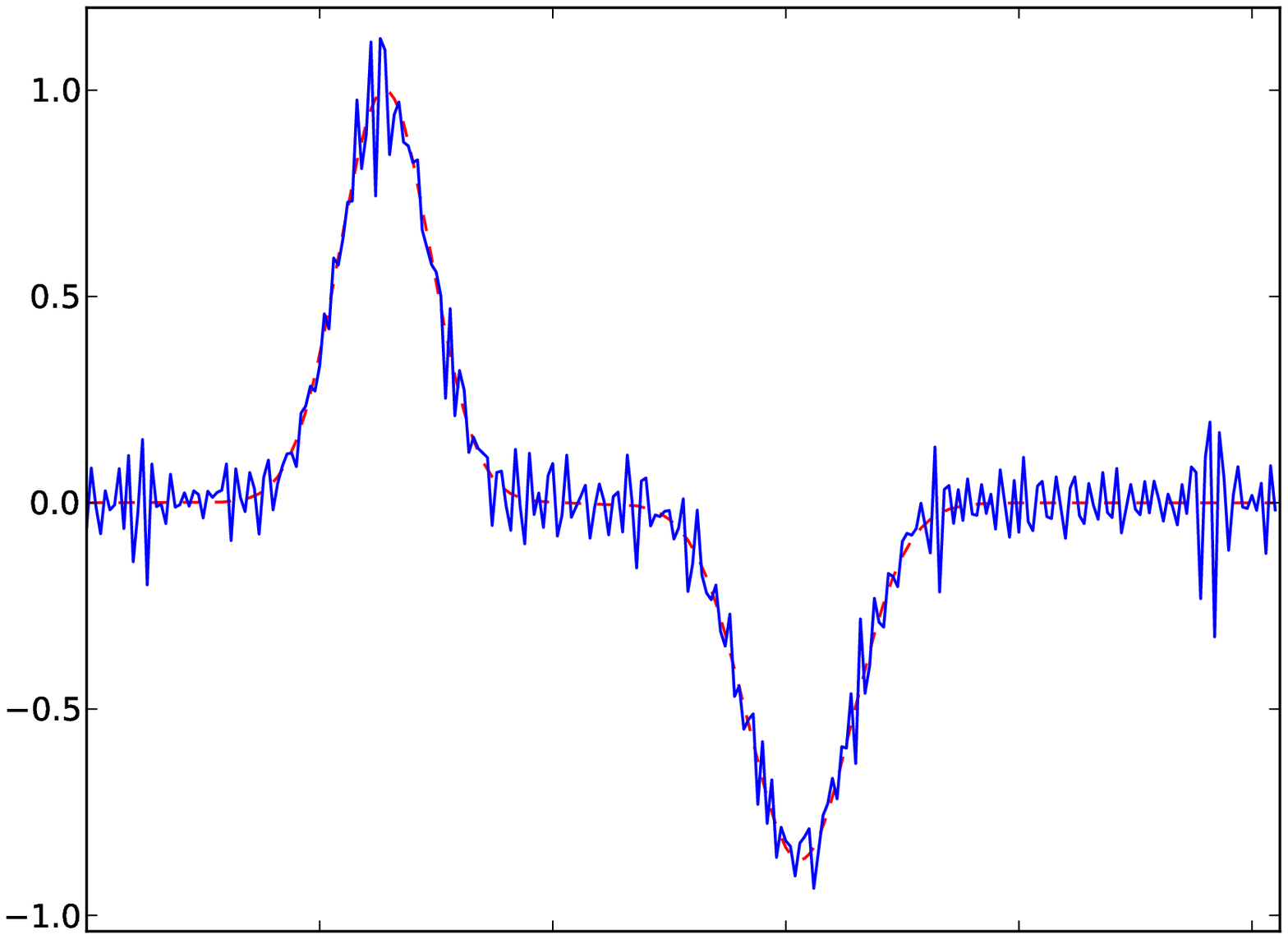}} \\       
  \subfloat{\includegraphics*[width=2.5in,trim=0 30 0 0,clip]{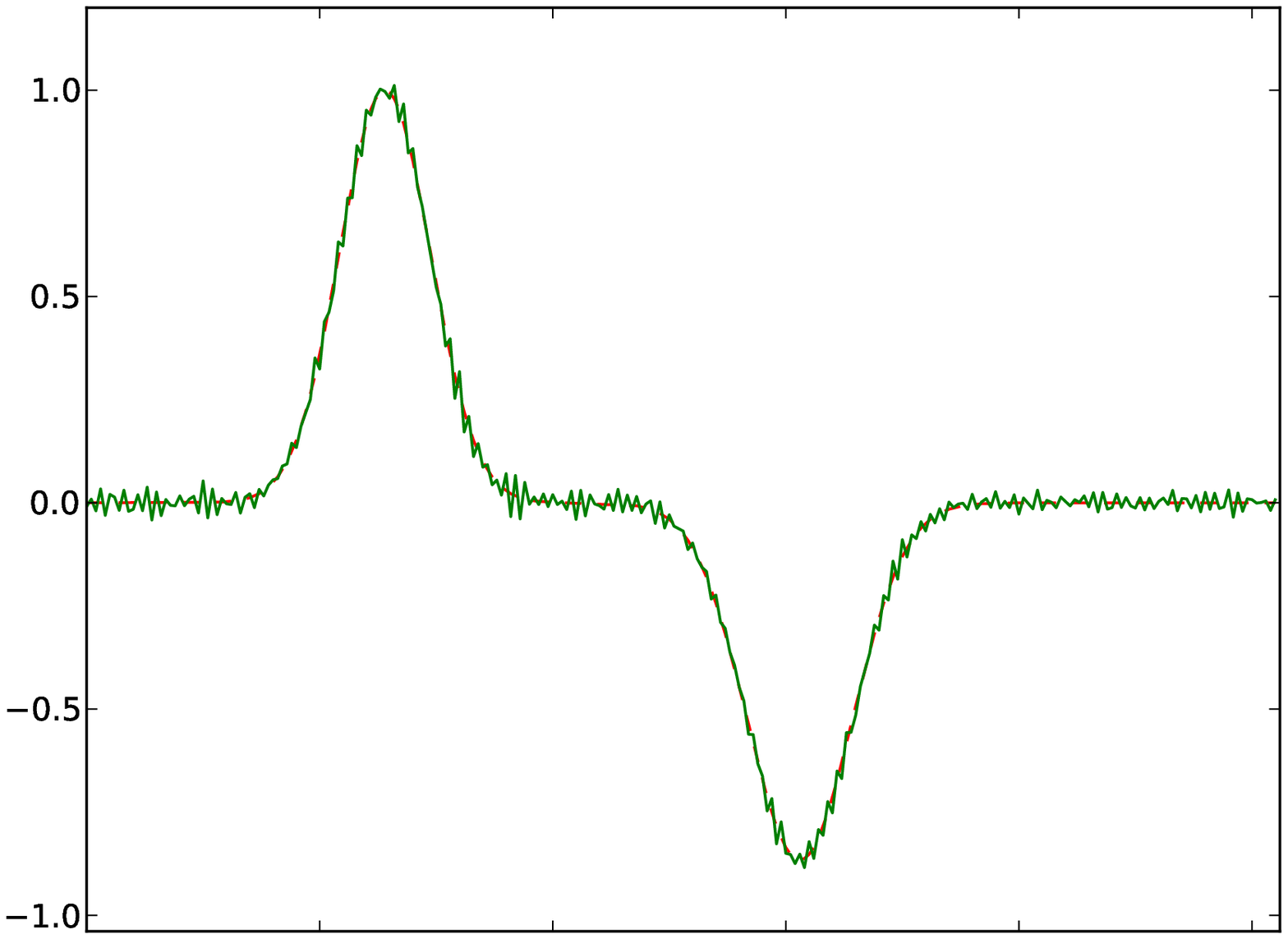}}
  \subfloat{\includegraphics*[width=2.5in,trim=0 30 0 0,clip]{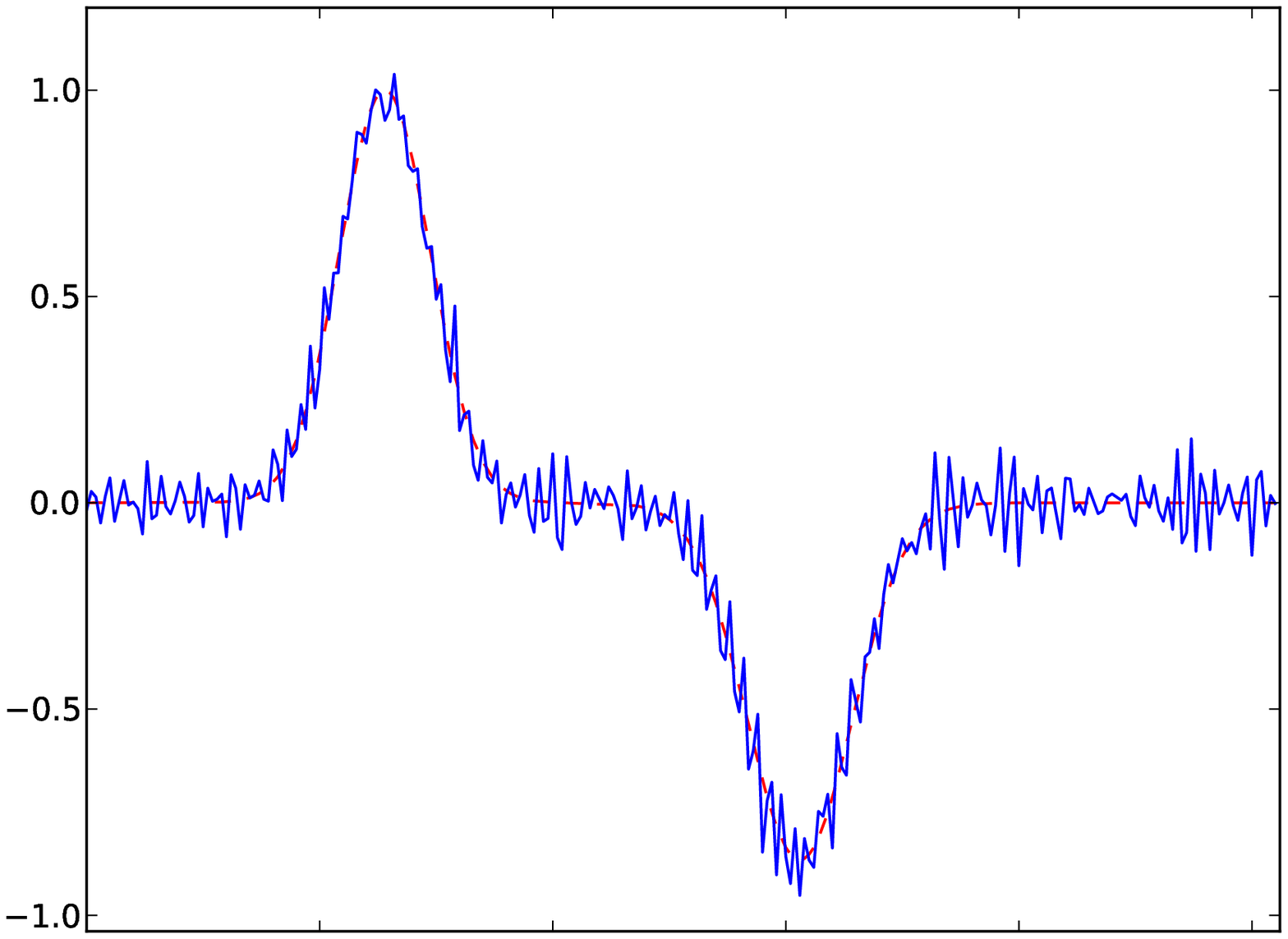}} \\       
  \subfloat[$\hat\theta_n$]{\includegraphics*[width=2.5in,trim=0 30 0 0,clip]{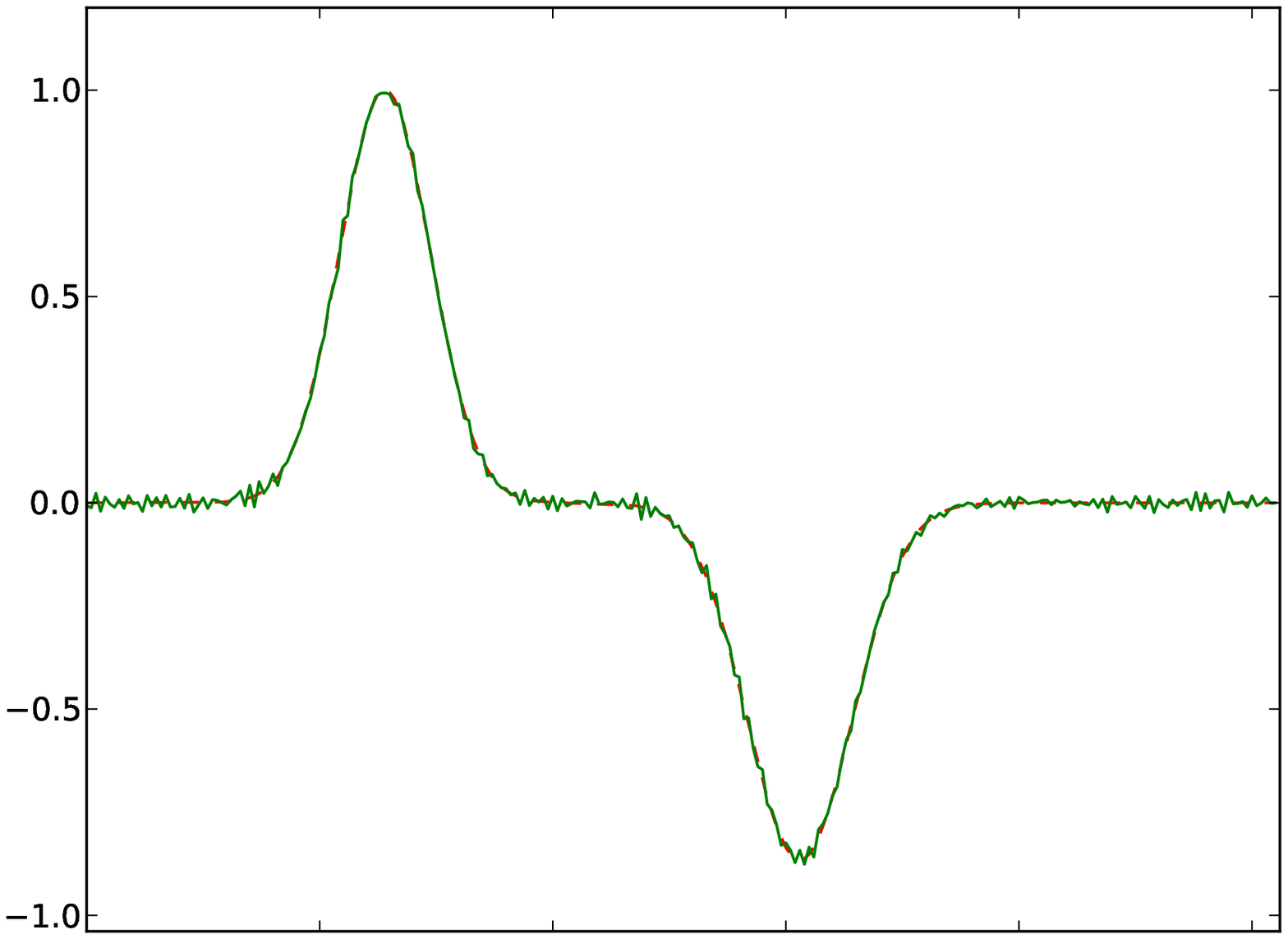}}
  \subfloat[$\hthetaRidge$]{\includegraphics*[width=2.5in,trim=0 30 0 0,clip]{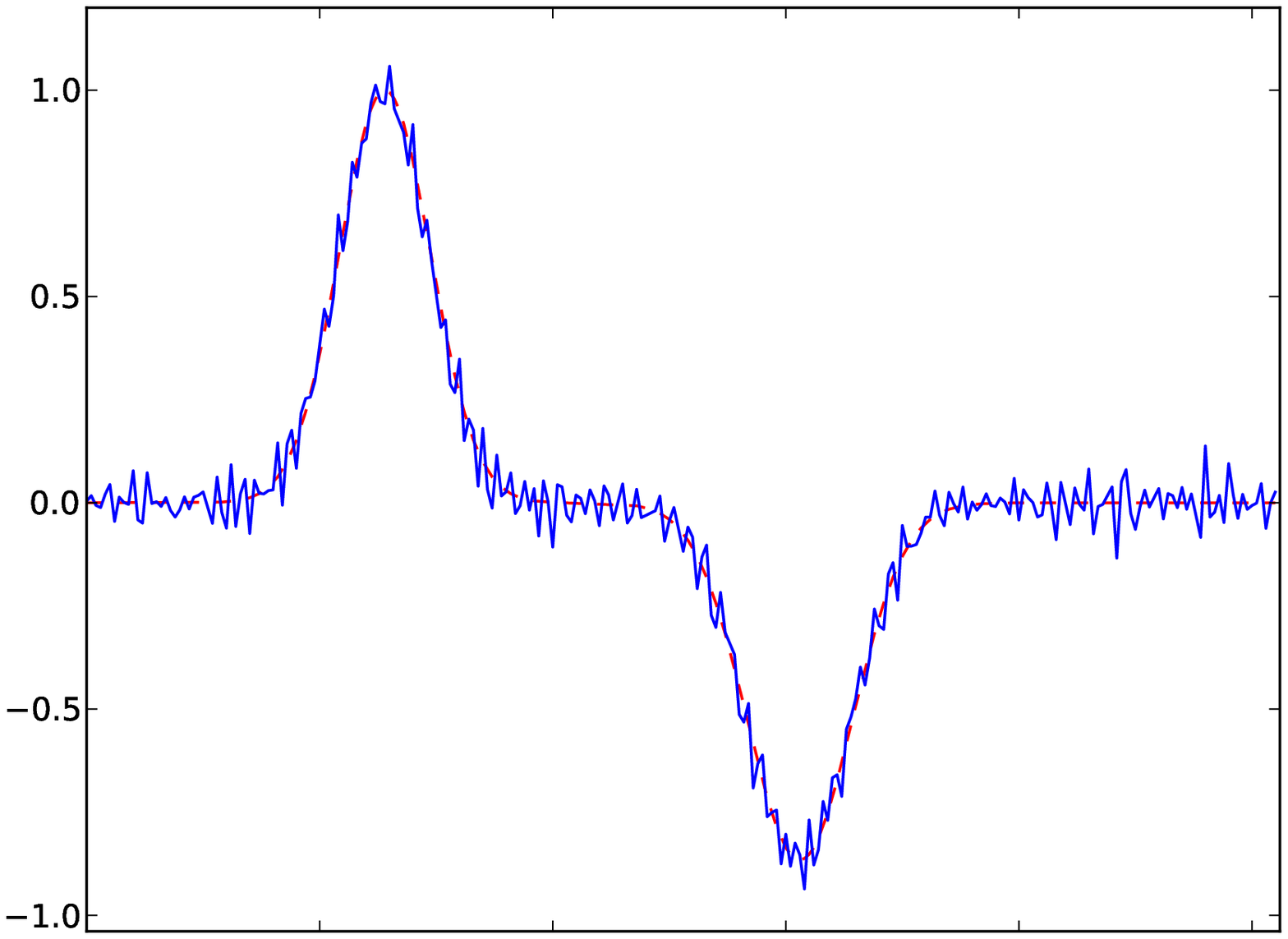}} \\
  \caption{Estimation of $\thetaS$ by $\hat\theta_n$ (left column) and $\hthetaRidge$ 
(right column).  The sample sizes range from
top to bottom, $n = 50, 100, 200, 300$.  Our estimator, $\hat\theta_n$,
quickly converges to $\thetaS$. However, $\hthetaRidge$, which doesn't zero out any
coefficients, still has substantial fluctuations after $n=300$ observations.  See Table \ref{tab:RR} 
for $RR$ results for this simulation.}
\label{fig:1dsmooth}
\end{figure}

For the signal 
$\theta^{\mathrm{peaked}}$, $\hat\theta_n$
estimates the true height of the peaks accurately and quickly.  In particular, the
secondary small peak is definitively identified with the correct shape and height 
for $n = 50$ observations, while for $\hthetaRidge$, the secondary peak is much less clear.
There are still some remaining oscillations at $n = 300$, resulting from 
unavoidable consequence of using the eigenvector basis.  This is a well-known
phenomenon in Fourier analysis known as the `Gibbs effect.'  Even with 
this obstacle, $\hat\theta_n$ converges quickly to $\theta^{\mathrm{peaked}}$.
See Figure \ref{fig:1dpeaked} for graphical results.

\begin{figure}          
  \centering            
  \subfloat{\includegraphics*[width=2.5in,trim=0 30 0 0,clip]{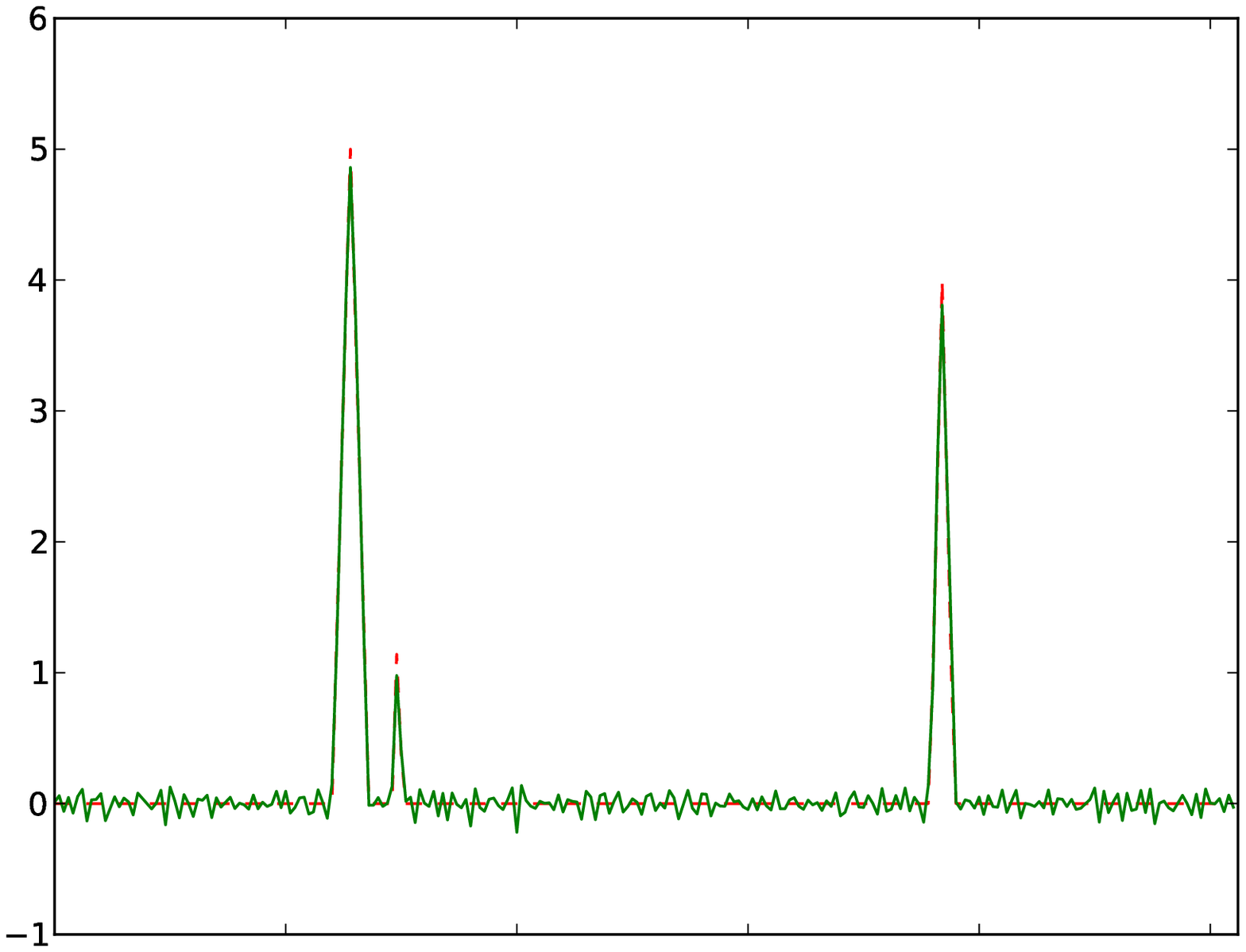}}
  \subfloat{\includegraphics*[width=2.5in,trim=0 30 0 0,clip]{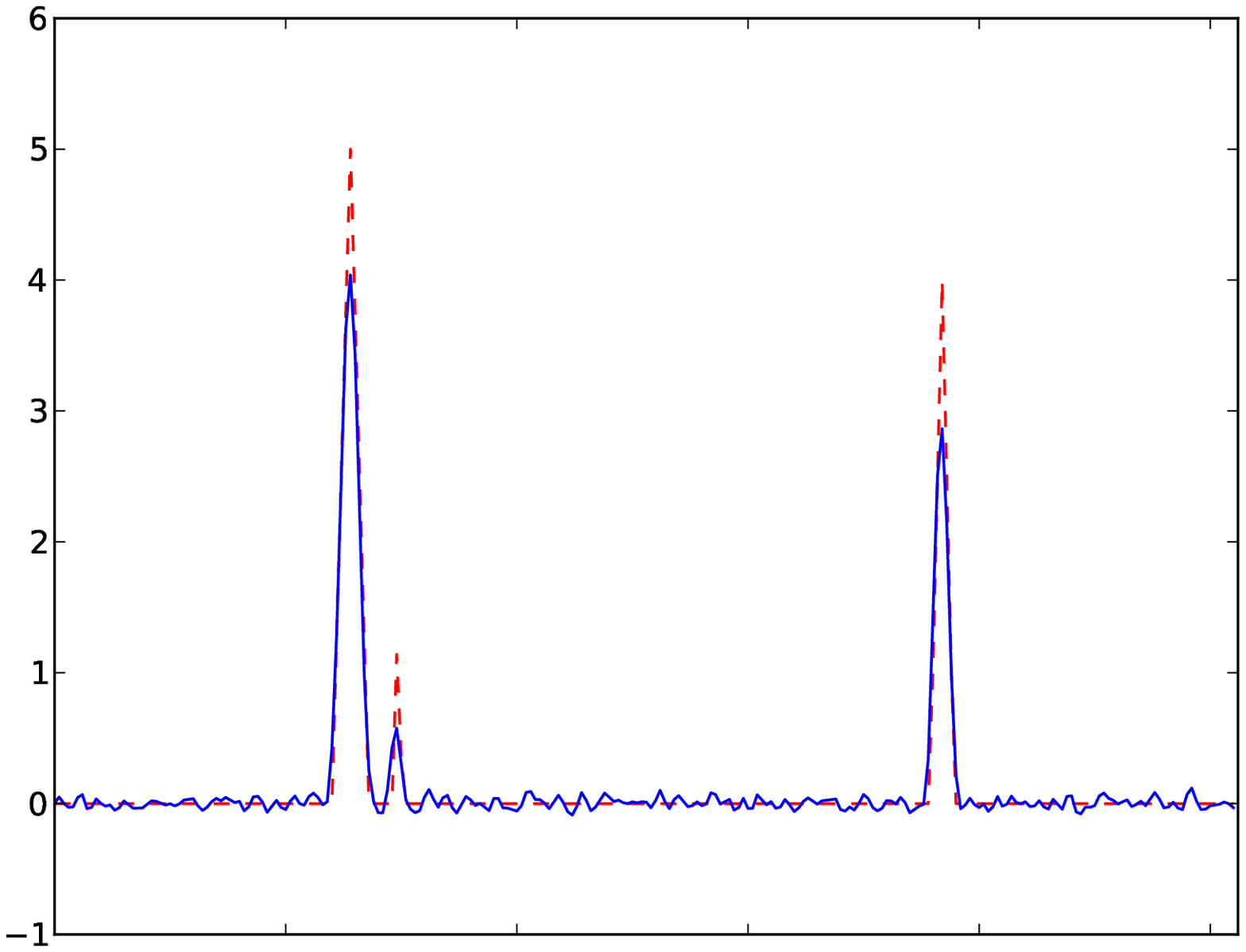}} \\       
  \subfloat{\includegraphics*[width=2.5in,trim=0 30 0 0,clip]{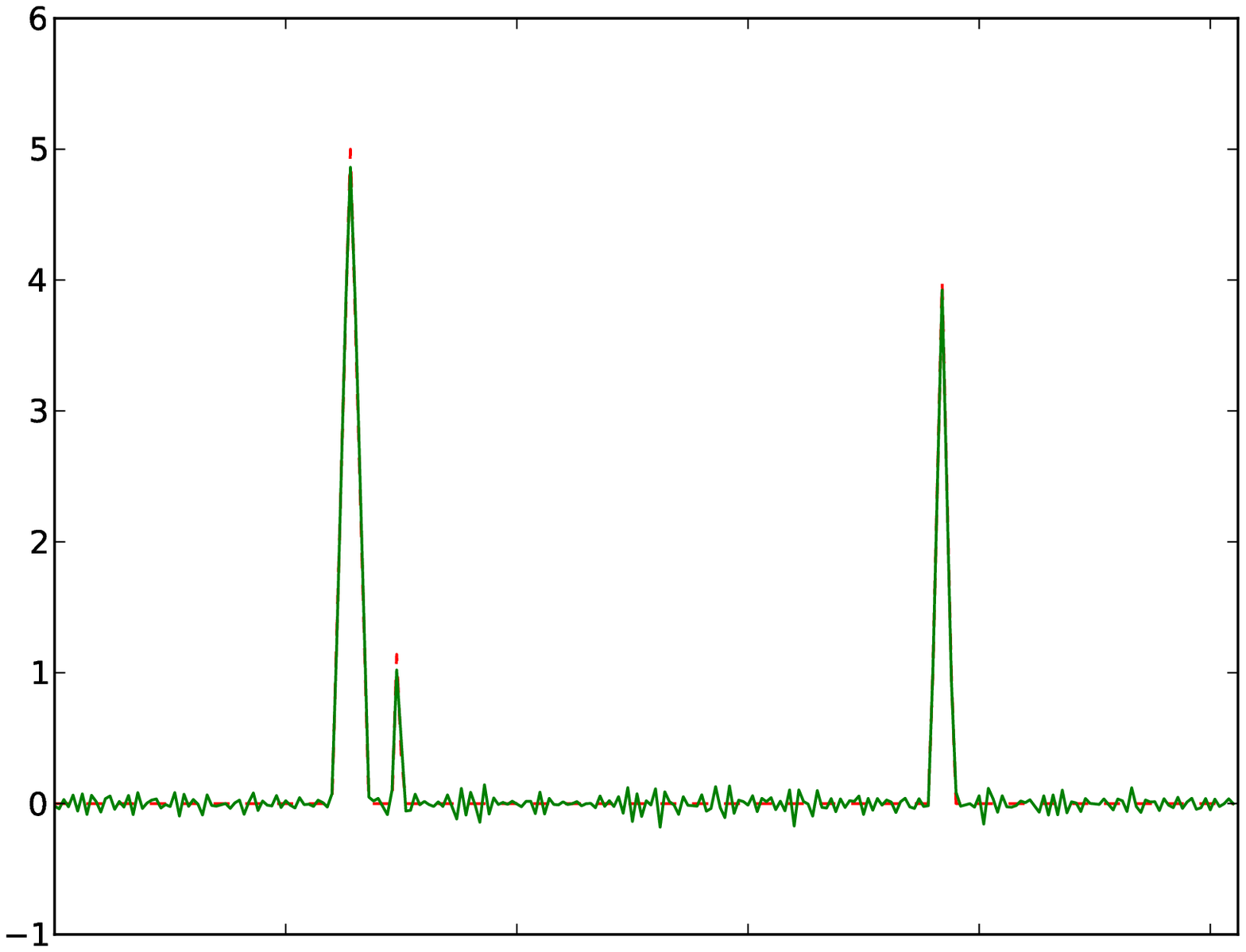}}
  \subfloat{\includegraphics*[width=2.5in,trim=0 30 0 0,clip]{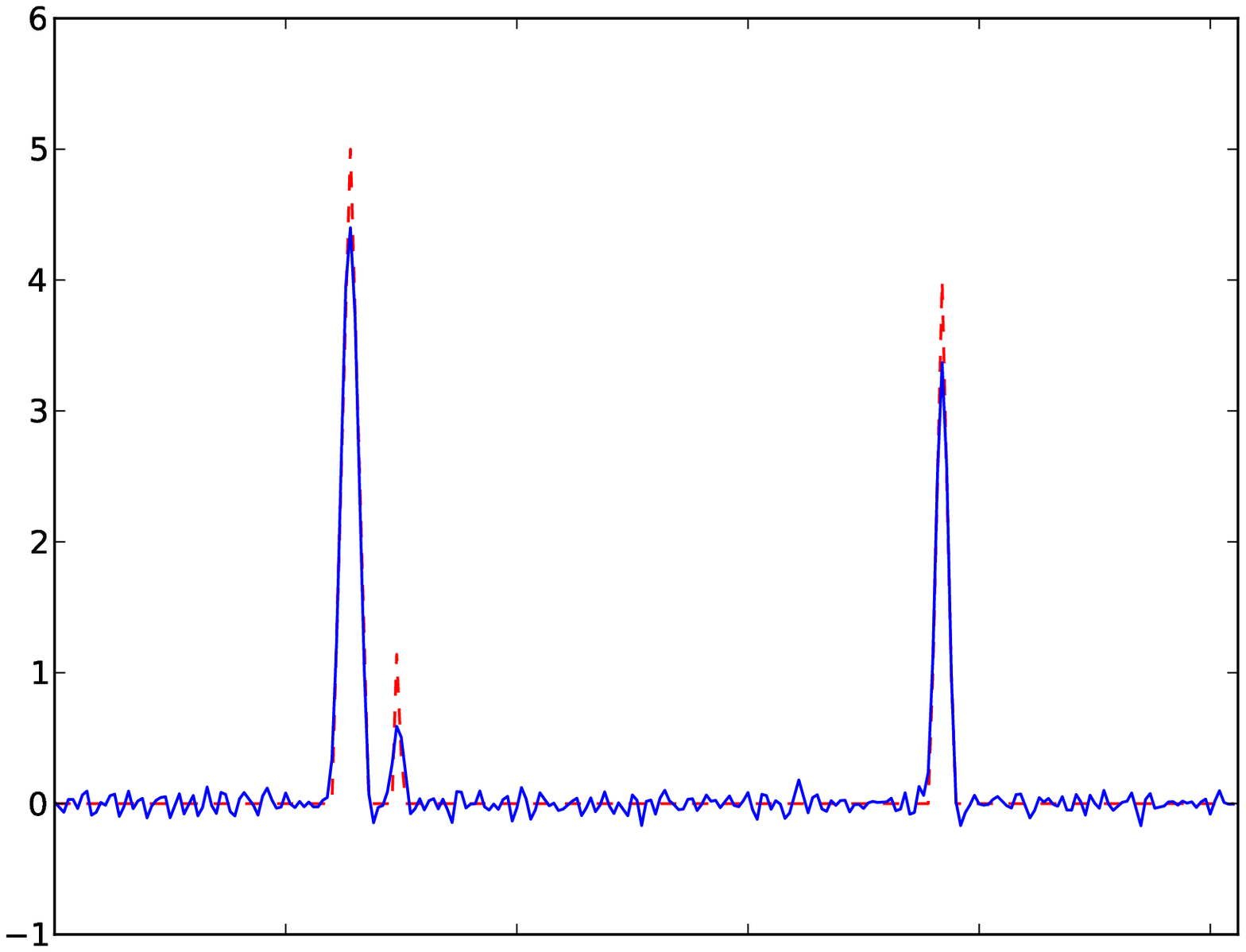}} \\       
  \subfloat{\includegraphics*[width=2.5in,trim=0 30 0 0,clip]{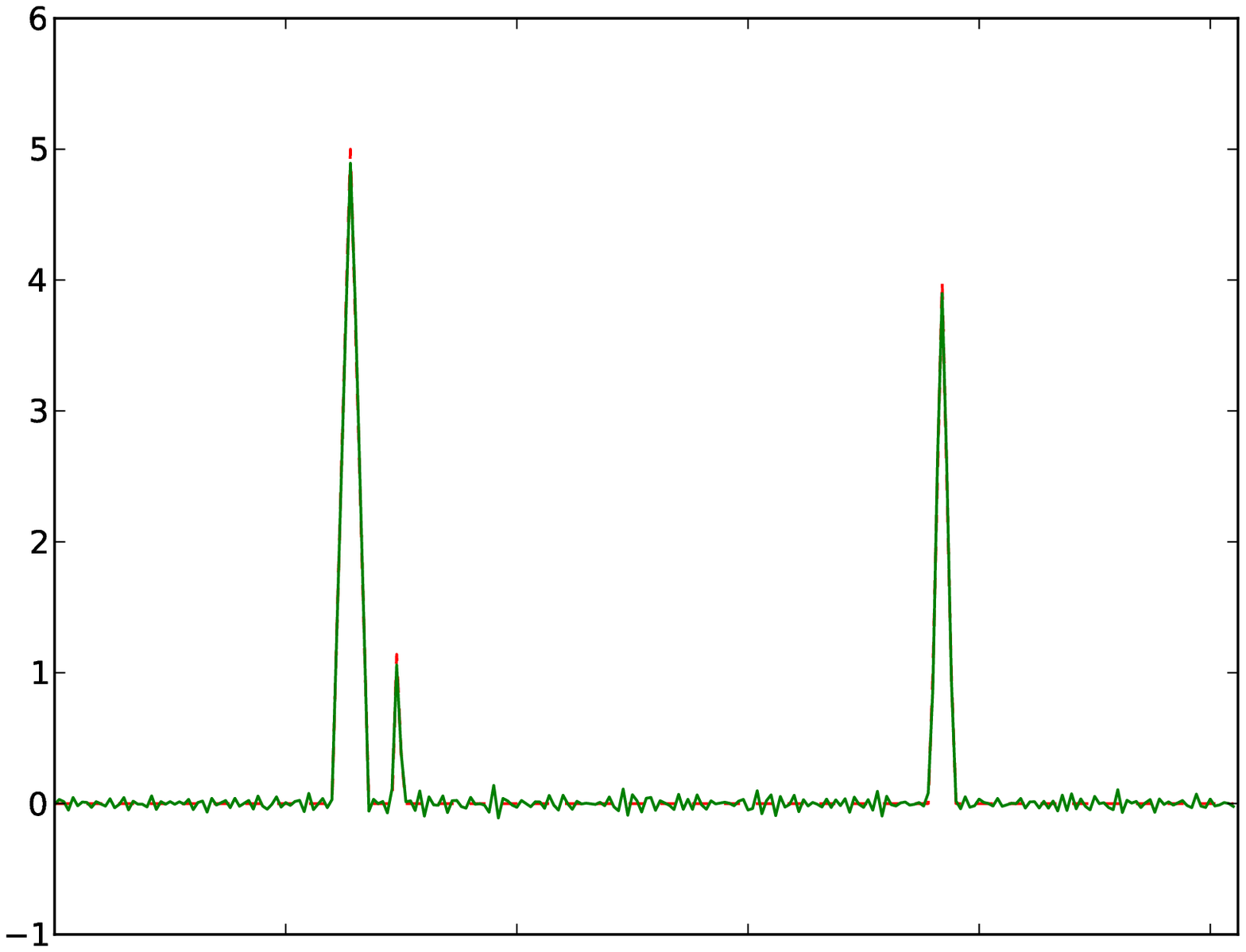}}
  \subfloat{\includegraphics*[width=2.5in,trim=0 30 0 0,clip]{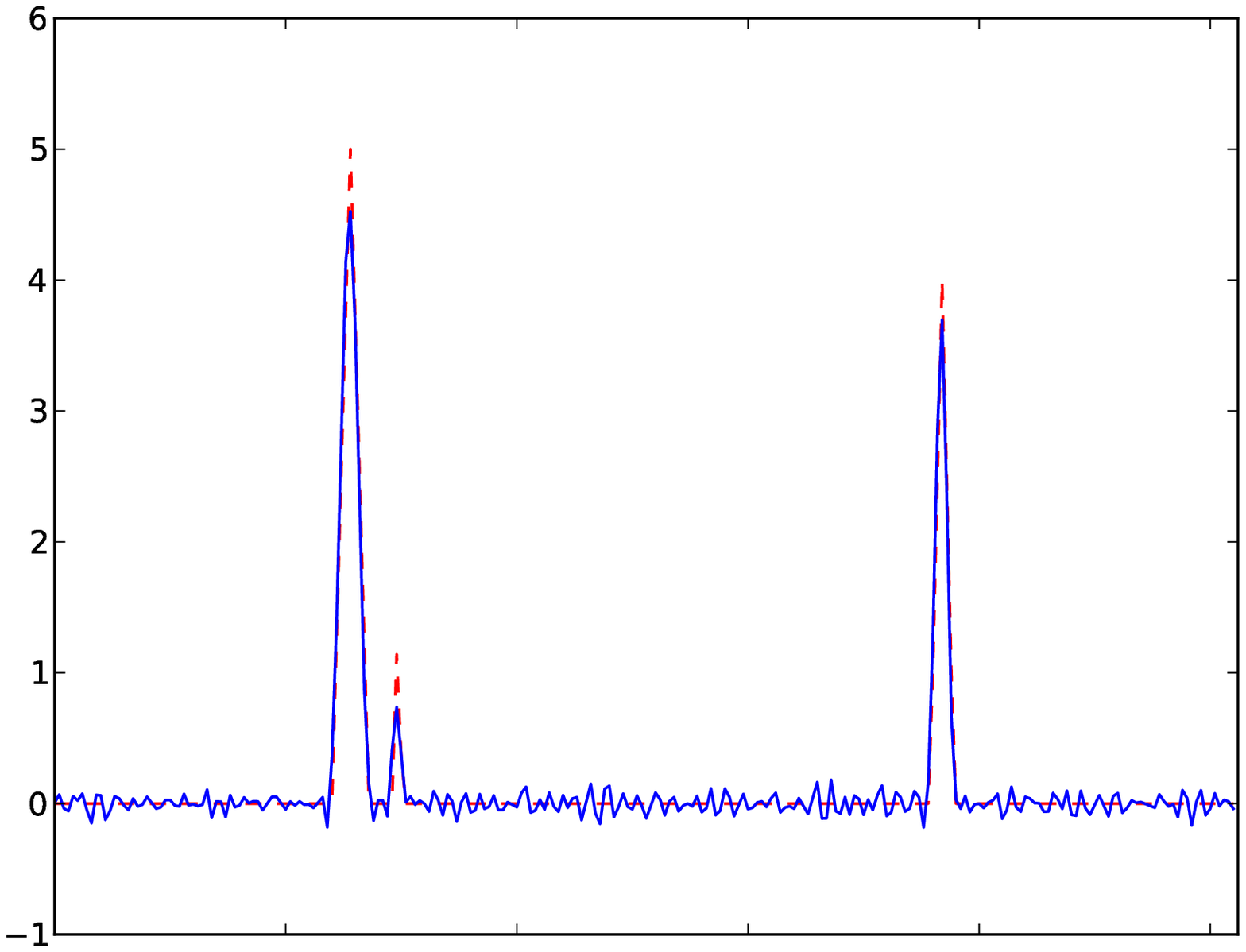}} \\       
  \subfloat[$\hat\theta_n$]{\includegraphics*[width=2.5in,trim=0 30 0 0,clip]{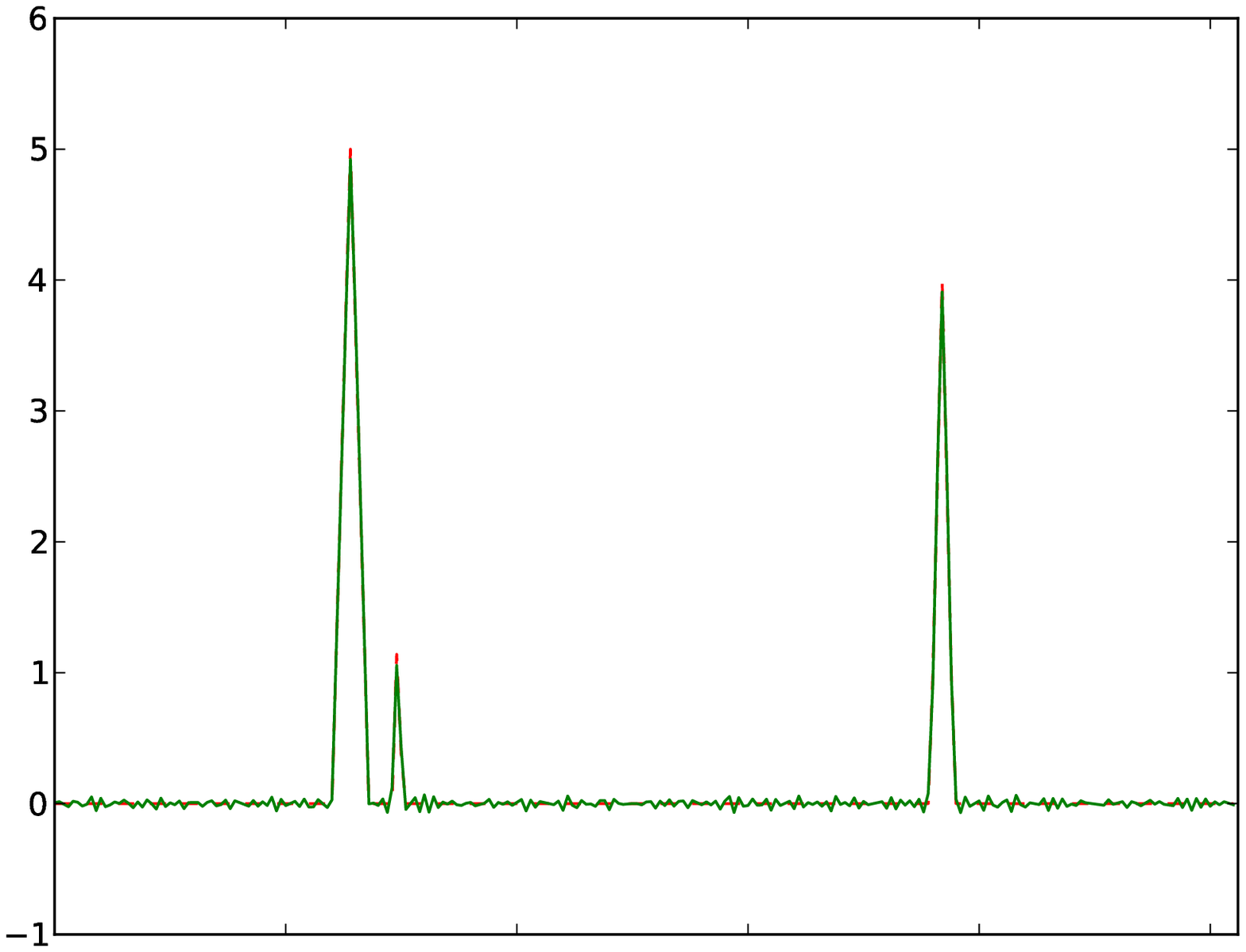}}
  \subfloat[$\hthetaRidge$]{\includegraphics*[width=2.5in,trim=0 30 0 0,clip]{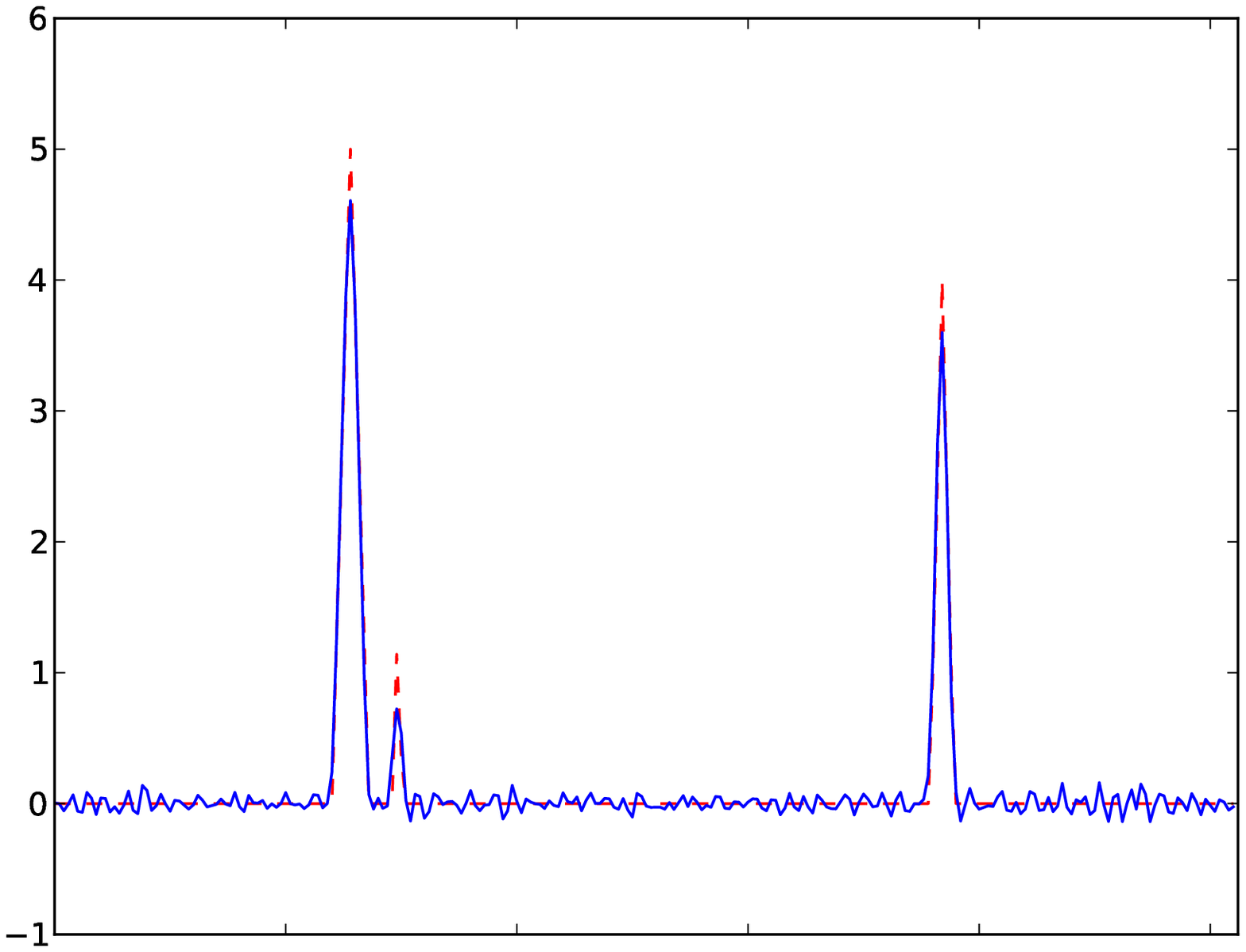}} \\
\caption{Estimation of $\theta^{\mathrm{peaked}}$ by $\hat\theta_n$ (left column) and $\hat \theta_{\text{ridge}}$ 
(right column).  The sample sizes range from
top to bottom, $n = 50, 100, 200, 300$.  Our estimator, $\hat\theta_n$,
estimates the true height of the peaks accurately and quickly.  In particular, the
secondary small peak is definitively identified with the correct shape and height.
There are still some remaining oscillations at $n = 300$, resulting from 
an unavoidable Gibbs effect from using the eigenvectors as a basis.
See Table \ref{tab:RR} 
for $RR$ results for this simulation.}
\label{fig:1dpeaked}
\end{figure}

\begin{table}[!h]
\begin{tabular}{l|cc|cc}
&&&& \\
  & $RR(\hat\theta_n,\thetaS)$ & $RR(\hthetaRidge,\thetaS)$
  & $RR(\hat\theta_n,\theta^{\text{peaked}})$ & $RR(\hthetaRidge,\theta^{\text{peaked}})$ \\
&&&& \\
\hline
$n = 50$              & 0.291              & 0.288       & 0.148              & 0.151 \\
$n = 100$             & 0.210              & 0.223       & 0.116              & 0.171 \\
$n = 200$             & 0.149              & 0.199       & 0.092              & 0.149 \\
$n = 300$             & 0.120              & 0.173       & 0.079              & 0.141   
\end{tabular} 
\caption[Results of $RR$ for $\thetaS$ and $\thetaP$]
{The $RR$ for the two considered simulations.  These are estimated by averaging 100 runs of our
simulations.}
\label{tab:RR}
\end{table}

\section{Discussion}
In this paper, we provide a general method for recovering an unknown signal
given a sequence of noisy observations that are only indirectly of that
signal of interest.  Our estimator, $\hat \theta_n$, has many favorable properties.  It has
computational efficiency in the sense that it can
be updated with a new observation without need to reference the entire
sequence of observations.  Instead, it relies on only a few summary statistics that need
to be maintained
and updated.  Though its computation is predicated on finding the eigenvectors
and eigenvalues of potentially large matrices, the implementation is straightforward
and generalizable to higher dimensional signals such as images.
Additionally, there exist accurate methods for the approximate computation of the eigenvectors
of matrices that could in principle be used to speed up the computation of $\Psi$.

Also, $\hat \theta_n$ is statistically efficient as well. The uniform consistency and oracle inequality
results show that it is making about as good a use of the data as possible.
Likewise, $\hat \theta_n$ has worked very well in our experiments so far, as evidenced 
by the results in Figures
\ref{fig:satelliteExample}, \ref{fig:1dsmooth}, and \ref{fig:1dpeaked}. Our estimator
can recover the unknown signal $\theta$ efficiently, requiring very few observations.
Also, referring to the second row of Figure \ref{fig:satelliteExample}, the collection
of a very poor observation merely doesn't improve the estimate instead of decreasing its quality.
This is in opposition to many currently implemented techniques such as straight averaging,
where low quality observations decrease the quality of the recovery.

\appendix

\section{}
\label{sec:sameSpectrum}
This section gives warrant for
assumption (A3) in Section \ref{sec:methodology}.
Although a  slightly weaker version of assumption (A3) is all that is actually
required (that  only the right eigenvectors need be the same
instead of both left and right eigenvectors) we leave it in its current form
for simplicity of exposition and conditions.

Two real matrices $A,B$ share the 
same eigenvectors if they are simultaneously unitarily diagonalizable;
that is, there exists two diagonal matrices $\Sigma_1,\Sigma_2$ and an unitary  matrix
$\Psi$ such that $A = \Psi\Sigma_1 \Psi^*$ and $B = \Psi\Sigma_2 \Psi^*$.  
Note $A$ and $B$ must of course be unitarily diagonalizable, which implies
by the spectral theorem that $A$ and $B$ are normal; that is $A^{\top}A = AA^{\top}$
and $B^{\top}B = BB^{\top}$.  The following
theorem characterizes simultaneous diagonalizability.

\begin{lemma}
Let $\mathcal{K}$ be a commuting family of normal matrices.  Then $\mathcal{K}$
is also simultaneously unitarily diagonalizable.
\label{lemma:sameEvecsLemma}
\end{lemma}

\begin{proof}[Proof of Lemma \ref{lemma:sameEvecsLemma}]
By the Schur unitary triangularization theorem \citep[Theorem 2.3.1]{horn1985}
if $\mathcal{K}$ is a commuting family of matrices, then there is a unitary $\Psi$
such that $\Psi K \Psi^*$ is upper triangular for every $K \in \mathcal{K}$.  Hence,
as normality is preserved under unitary congruence and a triangular normal matrix must
be diagonal, the result follows.
\end{proof}

Though all Toeplitz matrices commute asymptotically as the number of
rows and columns increases, not all Toeplitz matrices commute for a fixed size.
Many subsets of the family of Toeplitz matrices satisfy Lemma \ref{lemma:sameEvecsLemma}, 
however.  In particular, all circulant matrices commute \citep[Chapter 3.1]{gray2001}.  This shows
Theorem \ref{thm:sameEvecsMain}.

\section{}
\label{sec:proofs}
We utilize the following notation in several of the below proofs.  We use $\lesssim$ to indicate
`less than or equal to up to a constant independent of $n$.'
Also, it is convenient to think of a complex number $a = a_1 + a_2i$ as an element 
$(a_1,a_2) \in \mathbb{R}^2$.  
In this case, we use $|||a|||^2 = a_1^2 + a_2^2$ as a norm on $\mathbb{R}^2$, 
as the complex modulus is not technically defined
on elements of $\mathbb{R}^2$.  Additionally, 
$Z \sim N(0,I_2)$ is the two dimensional standard normal. Lastly, 
we define $s_{nj} := \Omega_n^2 \ep^2/\Delta_{nj}$.

We begin with a lemma that will be used in the proofs of Theorem \ref{thm:mainResult1} and
Theorem \ref{thm:mainResult2}:
\begin{lemma}
Let $\mu \in \mathbb{R}^2$ be a vector, 
$\Sigma = $ diag$(\sigma_1^2,\sigma_2^2)$ be a diagonal matrix with positive entries, and $c^2$ be 
a real, positive constant.  Then
\begin{equation}
\mathbb{P}( ||| \mu + \Sigma^{1/2} Z |||^2 \leq c^2 ) \leq 
\mathbb{P}( ||| \mu + \sigma_{\text{max}} Z |||^2 \leq c^2 )
\label{eq:useMaxVarLemma}
\end{equation}
if $|||\mu||| > c$ and $\sigma_{\text{max}} = \max\{\sigma_1,\sigma_2\}$.
\label{lemma:useMaxVar}
\end{lemma}
Here, we don't give a formal proof but provide intuition. The probability in equation 
\eqref{eq:useMaxVarLemma} 
corresponds to the amount of the mass of an elliptical normal, aligned with the cononical axis,
that resides in a ball of radius $c$ at the origin.
Hence, if $|||\mu||| > c$ (that is, the mean is outside the ball) 
a more spread out the normal results in more mass  inside the ball.  

\begin{proof}[Proof of Theorem \ref{thm:mainResult1}]
For simplicity, write $\hat\be_n := \hat \lamVec(\mathbf{B}_n)$.
Then
\[
\sup_{\theta \in \Theta} R_n(\hat \theta, \theta) = \sup_{\be \in \B} R_n(\hat \be_n, \be),
\]
where $\B := \{\be: ||\be||^2 \leq \parm \} = \Psi^*\Theta$.  Then we wish to show that
\begin{equation}
\limsup_{n \rightarrow \infty} \sup_{\be \in \B} R_n(\hat \be_n, \be) = 0,
\end{equation}
where the subscript $n$ on $R$ has been included to emphasize the dependence on $n$.

We begin by defining the following set
\[
A_j := \{\omega: |B_{nj}(\omega)|^2 > s_{nj}^2 \}
\]
where $\omega$ ranges over the measure space on which the random variable $B_{nj}$ is defined.
The utility of defining $A_j$ is
\begin{equation}
\hat \be_{nj} \mathbf{1}_{A_j} = 
\left(1 - \frac{\Omega_n^2\ep^2}{\Delta_{nj}|B_{nj}|^2} \right) B_{nj} \mathbf{1}_{A_j}
\end{equation}

Additionally, write $B_{nj} = \be_j + Z_{nj}$ as a mean term plus stochastic term,
where $Z_{nj}$ is the $j^{th}$ entry in the complex normal 
$\ep\Delta_n^{-1}\sum_i (D_i^* \Psi^* W_i)$.  Then the following bound 
on the $j^{th}$ term in the loss holds:

\begin{align}
|\hat \be_{nj} -  \be_j|^2 
& =   
\IAj|\hat \be_{nj} -  \be_j|^2 + \IAjc|\hat \be_{nj}-\be_j|^2  \nonumber \\
& =  
\IAj \left| \left( 1 - \frac{s_{nj}^2}{|B_{nj}|^2} \right) B_{nj} - \be_j \right|^2
      + \IAjc|\be_j|^2 \nonumber\\
& = 
\IAj \left| Z_{nj} - \left( \frac{s_{nj}^2(\be_j + Z_{nj})}{|\be_j + Z_{nj}|^2} \right) \right|^2
      + \IAjc|\be_j|^2 \label{eq:mainInequalityUniform} \\
& \leq  
\IAj\left( |Z_{nj}| + \frac{s_{nj}^2}{|\be_j + Z_{nj}|}\right)^2
 + \IAjc|\be_j|^2 \nonumber\\
& \leq  
\IAj\left( |Z_{nj}| +  s_{nj}\right)^2
 + \IAjc|\be_j|^2. \nonumber
\end{align}

To show that the expected value of the first term goes to zero in expectation, observe:

\begin{align*}
\mathbb{E}\IAj\left( |Z_{nj}| + s_{nj}\right)^2
& =
\mathbb{E}|Z_{nj}|^2 + 2 s_{nj}\mathbb{E}|Z_{nj}| + s_{nj}^2 \\
& \leq
\mathbb{E}|Z_{nj}|^2 + 2 s_{nj}\sqrt{\mathbb{E}|Z_{nj}|^2} + s_{nj}^2 \\
& =  
\frac{\ep^2}{\Delta_{nj}} + 2 s_{nj}\sqrt{\frac{\ep^2}{\Delta_{nj}}} + s_{nj}^2 \\
& \leq  
\frac{\ep^2}{\Delta_{nj}}\left(1 + 2\Omega_n + \Omega_n^2 \right).
\end{align*}
As $\Omega_n^2 < C < \infty$ for $n$ large enough for some $C$ by assumption (A5),
\begin{equation}
\mathbb{E}\IAj\left( |Z_{nj}| + s_{nj}\right)^2 = O(1/\Delta_{nj})
\label{eq:firstPart}
\end{equation}
uniformly in $\be$.

For the second term, $\IAjc|\be_j|^2$, we need to show 
\begin{equation}
\limsup_{n \rightarrow \infty} \sup_{\be \in \B} \sum_{j=1}^p \mathbb{P}(A_j^c)|\be_j|^2 = 0.
\label{eq:supExpectation}
\end{equation}

First, we compute the eigenvalue matrix $\Lambda_{nj}$ 
of the covariance matrix of $Z_{nj}$ as a vector in $\mathbb{R}^2$.  By
the properties of complex normals\footnote{Technically, this covariance matrix is off
by a constant, but this is not relevant for our current purposes.} 
\[
Z_{nj} \sim N\left( 
\left(
\begin{array}{c}
0 \\
0
\end{array}
\right),
\left(
\begin{array}{ll}
\ep^2/\Delta_{nj} & \Im C_{jj} \\
\Im C_{jj}& \ep^2/\Delta_{nj} 
\end{array}
\right)
\right)
\]
where $C_{jj}$ is the $j^{th}$ diagonal entry of the matrix
$\ep^2\Delta_n^{-1}\sum_i (D_i^* \Psi^* \overline{\Psi} \overline{D_i})\Delta_n^{-1}$.
Hence, the entries in $\Lambda_{nj}$ are $\lambda_{nj,1}^2 = \ep^2/\Delta_{nj} + \Im C_{jj}$
and $\lambda_{nj,2}^2 = \ep^2/\Delta_{nj} - \Im C_{jj}$, which are both
strictly positive. Also, define $U$ to be the associated eigenvector matrix.

Though it is clear that $\mathbb{P}(A_j^c)|\be_j|^2$ goes to zero pointwise, 
the worst $\be_j$ is arbitarily close to zero.  Hence, to show uniform convergence,
we define a parameter $\tau_{nj}^2$. For each $j$, define
$\B_j := \{\be_j : |||\be_j|||^2 \leq \parm \}$ and split this set into
$\B_j = \B_{jn} \cup \B_{jn}^c$, where 
\[
\B_{jn}^c := \{\be : \tau_{nj}^2 \leq |||\be_j|||^2 \leq \parm\}.
\]
Also, as $|||\cdot|||$ is invariant under orthogonal operations, 
we can rotate everything by the eigenvectors $U$.  Denote rotation by
$U$ by a tilde; that is, $\tilde\be_j := U \be_j$.  Then,

\begin{align*}
\sup_{\be \in \B} \sum_{j=1}^p \mathbb{P}(A_j^c)|\be_j|^2 
& \leq \sum_{j=1}^p \sup_{\be_j \in \B_j} \mathbb{P}(A_j^c) |\be_j|^2 \\
& \leq 
\sum_{j=1}^p \max\left\{\sup_{\be_j \in \B_{nj}} \mathbb{P}_{\be_j}(A_j^c)|\be_j|^2
\sup_{\be_j \in \B_{nj}^c} \mathbb{P}(A_j^c)|\be_j|^2 \right\}\\
& \leq 
\sum_{j=1}^p \max\left\{\tau_{nj}^2,
\sup_{\be_j \in \B_{nj}^c} \mathbb{P}(A_j^c)|\be_j|^2 \right\}\\
& =
\sum_{j=1}^p \max\left\{\tau_{nj}^2,
\sup_{\be_j \in \B_{nj}^c} \mathbb{P}(|||U(\be_j + Z_n )|||^2 \leq s_{nj}^2)|||\tilde \be_j|||^2 \right\}. \\
\end{align*}

Then continuing on with the second term in the max,
and using Lemma \ref{lemma:useMaxVar}  with $|||\tilde\be_j||| > s_{nj}$, 
which happens if $\tau_{nj}^2 > s_{nj}^2$,

\begin{align}
& \sup_{\tilde \be_j \in \B_{nj}^c}\mathbb{P}(|||\tilde \be_j + \Lambda_{nj}^{1/2} Z |||^2 \leq s_{nj}^2)|||\tilde \be_j|||^2\nonumber\\
& \leq
\sup_{\tilde \be_j \in \B_{nj}^c} \mathbb{P}(|||\tilde \be_j + \lambda_{\text{max}} Z |||^2 \leq s_{nj}^2)|||\tilde \be_j|||^2\nonumber\\
& \leq
\sup_{\tilde \be_j \in \B_{jn}^c} \left(1 - \Phi\left(|||\tilde \be_j / \lambda_{\text{max}} ||| - 
s_{nj}/\lambda_{\text{max}}\right)\right)|||\tilde \be_j|||^2  \nonumber \\
& =
\sup_{\tau_{nj}^2 \leq u^2 \leq T^2} \left(1 - \Phi\left(1/\lambda_{\text{max}}(u -s_{nj})\right)\right)u^2  \nonumber \\
& =
\sup_{\frac{\tau_{nj}}{\lambda_{\text{max}}} - s_{nj} \leq t \leq \frac{T}{\lambda_{\text{max}}} - s_{nj}} 
(1 - \Phi( t )) (\lambda_{\text{max}}(t + s_{nj}))^2  \nonumber \\
& =
\lambda_{\text{max}}^2 \sup_{\frac{\tau_{nj}}{\lambda_{\text{max}}} - s_{nj} \leq t \leq \frac{T}{\lambda_{\text{max}}} - s_{nj}} 
(1 - \Phi( t )) (t + s_{nj})^2  \nonumber \\
& \leq
\lambda_{\text{max}}^2 \sup_{0 \leq t \leq \infty} 
(1 - \Phi( t )) (t + 1)^2  \qquad \text{for $n$ large enough}\nonumber \\
& \leq 
\lambda_{\text{max}}^2  \nonumber
\end{align} 
The last inquality follows by Figure \ref{fig:proofPlot} and by noting that $(1 - \Phi( t )) (t + 1)^2$ is
continuous, unimodal, and bounded by 1.
\begin{figure}[!h]
\includegraphics[width=2.5in]{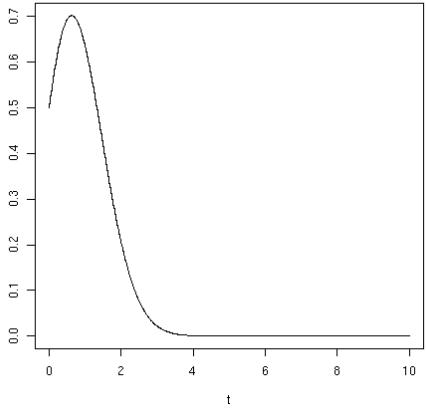}
\caption{Plot of $(1 - \Phi( t )) (t + 1)^2$.}
\label{fig:proofPlot}
\end{figure}

Therefore
\begin{equation}
\sup_{\be \in \B}  \mathbb{P}(A_j^c)|\be_j|^2 = \max\{ \tau_n^2, \lambda_{\text{max}}^2 \}
\end{equation}
Hence, it is sufficient to choose $\tau_{nj}^2 = 2 s_{nj}^2$ and to note that
\[
\lambda_{\text{max}}^2 \asymp s_{nj}^2 \asymp \ep^2/\Delta_{nj}.
\]  
This implies
\begin{equation}
\sup_{\be \in \B}  \mathbb{P}(A_j^c)|\be_j|^2 = O(s_{nj}^2).
\label{eq:supExpectation}
\end{equation}
As we are summing over $j$ in the risk, we conclude that 
\[
\limsup_{n \rightarrow \infty}\sup_{\be \in \B} \, \gamma_n^{-1}
R(\hat \beta_n,\be) < \infty
\]
where 
\[
\gamma_n = \max_j \frac{\ep^2}{\Delta_{nj}}.
\]

\end{proof}

\begin{proof}[Proof of Theorem \ref{thm:mainResult2}]
We use the same notations and conventions as in the proof of Theorem \ref{thm:mainResult1}.
Note that if we define $\sigma_{nj}^2 = \ep^2/\Delta_{nj}$, then the linear oracle risk is 
\begin{equation}
R(\be_*,\be) = \min_{\tilde \be = \lamVec(\mathbf{B}_n)} R(\tilde\be,\be) 
= \sum_{j=1}^p \frac{|\be_j|^2\sigma_{nj}^2}{\sigma_{nj}^2 + |\be_j|^2}
= \sum_{j=1}^p \frac{|\be_j|^2s_{nj}^2}{s_{nj}^2 + \Omega_n^2|\be_j|^2}.
\end{equation}

we bound the $j^{th}$ term in the loss:
\begin{align}
\lefteqn{|\hat\be_j - \be|^2} \nonumber\\ 
& = \IAj \left[ |Z_{nj}|^2 - 
            \frac{\overline{Z_{nj}}s_{nj}^2(\be_j + Z_{nj})}{|\be_j + Z_{nj}|^2}  -
            \frac{Z_{nj} s_{nj}^2\overline{(\be_j + Z_{nj})}}{|\be_j + Z_{nj}|^2}  + 
            \frac{s_{nj}^4}{|\be_j + Z_{nj}|^2}  \right] + \nonumber \\
& \qquad \qquad 
+ \IAjc|\be_j|^2 \nonumber \\
& = \IAj \left[ |Z_{nj}|^2 - 
           \frac{|Z_{nj}|^2s_{nj}^2}{|\be_j + Z_{nj}|^2}  -
           \frac{s_{nj}^2(|Z_{nj}|^2 + \be_j\overline{Z_{nj}} + \overline{\be_j}Z_{nj})}
                {|\be_j + Z_{nj}|^2}  + 
            \frac{s_{nj}^4}{|\be_j + Z_{nj}|^2}  \right] + \nonumber \\
& \qquad \qquad 
+ \IAjc|\be_j|^2 \nonumber \\
& = 
\IAj \left[ |Z_{nj}|^2 - \left(\frac{s_{nj}^2(|\be_j + Z_{nj}|^2 - |\be_j|^2)}
                                          {|\be_j + Z_{nj}|^2} \right) +
               \frac{s_{nj}^4}{|\be_j + Z_{nj}|^2}  \right] +
           \IAjc|\be_j|^2 \nonumber \\
& = 
\IAj \left[ |Z_{nj}|^2 - s_{nj}^2 + \left(\frac{s_{nj}^2|\be_j|^2}
                                          {|\be_j + Z_{nj}|^2} \right) +
               \frac{s_{nj}^4}{|\be_j + Z_{nj}|^2}  \right] +
           \IAjc|\be_j|^2 \nonumber \\
& \leq 
|Z_{nj}|^2 + \IAj\left(\frac{s_{nj}^2|\be_j|^2}{|\be_j + Z_{nj}|^2} \right)+\IAjc|\be_j|^2. \nonumber \\
& =
|Z_{nj}|^2 + 
\IAj\left(\frac{s_{nj}^2|\be_j|^2}{s_{nj}^2 + \Omega_n^2|\be_j|^2} \right)
\left(\frac{s_{nj}^2 + \Omega_n^2|\be_j|^2}{|\be_j + Z_{nj}|^2} \right) 
+ \IAjc|\be_j|^2. \nonumber
\label{eq:main}
\end{align}

By the previous proof, we see that the expected value of the first and third term go to zero
uniformly over $\be \in \B$ at rate $O(1/\Delta_{nj})$; the same rate as the oracle. 
For the second term, notice that 
\[
\IAj\left(\frac{s_{nj}^2 + \Omega_n^2|\be_j|^2}{|\be_j + Z_{nj}|^2} \right)
\leq 
\IAj\left(1 + \frac{\Omega_n^2|\be_j|^2}{|\be_j + Z_{nj}|^2} \right)
\lesssim 
\frac{\IAj|\be_j|^2}{|\be_j + Z_{nj}|^2} 
=: G_{nj}
\]
for $n$ large enough, by assumption (A5). Then our goal is to show that 
\[
\limsup_{n\rightarrow\infty} \sup_{\be\in\B} \mathbb{E}G_{nj} < \infty.
\]

First, due to $G_{nj}$ being rotationally symmetric (once we use $|||\cdot|||$
instead of $|\cdot|$), we
renormalize to transform $Z_{nj}$ into a vector $\tZ$ with independent 
standard normal components
\begin{align*}
G_{nj} & = \IAj\left(\frac{|||\be_j|||^2}{|||\be_j + Z_{nj}|||^2} \right) \\
& = \IAjt\left(\frac{|||U^{\top}\be_j|||^2}{|||U^{\top}\be_j + U^{\top}Z_{nj}|||^2} \right) \\
& = \IAjt\left(\frac{|||\tilde\be_j|||^2}{|||\tilde\be_j + \Lambda_{nj}^{1/2}\tZ|||^2} \right).
\end{align*}
We define $\Lambda_{nj}$ and $U$ in the previous proof
and $\tilde A_{j} := \{||\tilde\be_j + \Lambda_{nj}^{1/2}\tZ||^2 > s_{nj}^2$.

We break bounding $\mathbb{E}G_{nj}$ into cases. \\

\noindent \textbf{Case 1:} $|||\tilde \be_j|||^2 \leq s_{nj}^2$ \\
We see from the definition of $\IAjt$
\begin{equation}
G_{nj} 
\leq 
\IAjt\left(\frac{s_{nj}^2}{|||\tilde \be + \Lambda_{nj}\tZ |||^2} \right)
<
\IAjt\left(\frac{s_{nj}^2}{s_{nj}^2} \right) 
\leq 1
\end{equation}

\noindent \textbf{Case 2:} $|||\tilde \be|||^2 > s_{nj}^2$ \\
Note that by the nonnegativity of $G_{nj}$
\[
\mathbb{E} G_{nj} = \int_0^\infty \mathbb{P}(G_{nj} > \tau)\, d\tau.
\]
As an aside, the random variables considered don't put any positive mass at any points, 
so we don't need to worry about whether the boundaries of integration are included. For $\tau > 0$,
\begin{align*}
\mathbb{P}(G_{nj} > \tau) 
& = 
\mathbb{P}
\left( s_{nj}^2 \leq |||\tilde \be + \Lambda_{nj}^{1/2} \tZ|||^2 < \frac{|||\tilde\be|||^2}{\tau}\right) \\
& = 
\begin{cases}
0  & \tau \geq \frac{|||\tilde \be|||^2}{s_{nj}^2} \\
\mathbb{P}
\left( s_{nj}^2 \leq |||\tilde \be + \Lambda_{nj}^{1/2} \tZ|||^2 < \frac{|||\tilde\be|||^2}{\tau}\right) 
   & \text{o.w.}
\end{cases}
\end{align*}

Therefore, for any $c^2 > 0$,
\begin{align*}
\mathbb{E} G_{nj} 
& = 
\int_0^\infty \mathbb{P}(G_{nj} > \tau) \,d\tau \nonumber\\
& =
\int_0^{c^2}  \mathbb{P}(G_{nj} > \tau) \,d\tau  +  
\int_{c^2}^{\frac{|||\tilde \be|||^2}{s_{nj}^2}} \mathbb{P}(G_{nj} > \tau) \, d\tau \\
& \leq
c^2 + \left(\frac{|||\tilde \be|||^2}{s_{nj}^2}\right) \mathbb{P}(G_{nj} > c^2) \nonumber\\
& \leq
c^2 + \left(\frac{|||\tilde \be|||^2}{s_{nj}^2}\right)
\mathbb{P}
\left( |||\tilde \be + \Lambda_{nj}^{1/2} \tZ|||^2 < \frac{|||\tilde\be|||^2}{c^2}\right)  \nonumber\\
\end{align*}
If $c^2 > 1$, then the mean of
of the random variable $\tilde \be + \Lambda_{nj}^{1/2} \tZ$ will be outside of the circle
centered at zero with radius $||\tilde \be||/c$. Hence, by Lemma \ref{lemma:useMaxVar}
if we define $\lamMax^2  := \max\{\text{diag}(\Lambda_{nj})\}$, then it follows that
\begin{equation}
\mathbb{P}
\left( |||\tilde \be + \Lambda_{nj}^{1/2} \tZ|||^2 < \frac{|||\tilde\be|||^2}{c^2}\right)
\leq
\mathbb{P}
\left( |||\tilde \be + \lamMax\tZ|||^2 < \frac{|||\tilde\be|||^2}{c^2}\right).
\end{equation}
Using this, observe
\begin{align*}
\lefteqn{\left(\frac{|||\tilde \be|||^2}{s_{nj}^2}\right)
\mathbb{P}
\left( |||\tilde \be + \Lambda_{nj}^{1/2} \tZ|||^2 < \frac{|||\tilde\be|||^2}{c^2}\right)} \\
& \leq
\left(\frac{|||\tilde \be|||^2}{s_{nj}^2}\right)
\mathbb{P}
\left( |||\tilde \be/\lamMax + \tZ|||^2 < \frac{|||\tilde\be/\lamMax|||^2}{c^2}\right) \\
& \leq 
\left(\frac{|||\tilde \be|||^2}{s_{nj}^2}\right)
\left(1-\Phi\left( \left(1-\frac{1}{c}\right) |||\tilde \be/\lamMax|||\right)\right) \\
& =
\left(\frac{|||\tilde \be|||^2}{s_{nj}^2}\right)
\left(1-\Phi\left( \left(1-\frac{1}{c}\right) |||\tilde \be/\lamMax|||\right)\right) \\
& =
\left(\frac{(\lamMax t)^2}{s_{nj}^2}\right)
\left(1-\Phi\left( \left(1-\frac{1}{c}\right) t\right)\right) \\
& =
\left(\frac{\lamMax^2}{s_{nj}^2}\right) \left[t^2
\left(1-\Phi\left( \left(1-\frac{1}{c}\right) t\right)\right)\right] \\
\end{align*}
Where we have transformed $t = |||\tilde \be|||/\lamMax$.  Hence, as $s_{nj}^2 \asymp \lamMax^2$
and
\[
\sup_{\frac{s_{nj}}{\lamMax} \leq t \leq \frac{T}{\lamMax}} t^2
\left(1-\Phi\left( \left(1-\frac{1}{c}\right) t\right)\right)
\leq
\sup_{0 \leq t \leq \infty} t^2
\left(1-\Phi\left( \left(1-\frac{1}{c}\right) t\right)\right)
\leq 1
\]
we see that
\[
\left(\frac{|||\tilde \be|||^2}{s_{nj}^2}\right)
\mathbb{P}
\left( |||\tilde \be + \Lambda_{nj}^{1/2} \tZ|||^2 < \frac{|||\tilde\be|||^2}{c^2}\right)
= O(1),
\]
independent of $\be$.  And we conclude that
\[
\mathbb{E} G_{nj} = O(1),
\]
again, independent of $\be$.  This ends the proof.

\end{proof}


\begin{proof}[Proof of Theorem \ref{thm:randomEvalsUniformlyConsistent}]
Observe
\begin{align}
\nonumber \lim_{n\rightarrow \infty} \sup_{\be \in \B} \mathbb{E}_{(D_i),X_n} ||\hat \be - \be||^2 
& = 
\lim_{n\rightarrow \infty}\sup_{\be \in \B}\mathbb{E}_{(D_i)}\mathbb{E}_{X_n|(D_i)}||\hat \be - \be||^2 \\
& \leq 
\lim_{n\rightarrow\infty } \mathbb{E}_{(D_i)} \sup_{\be \in \B} R(\hat\be,\be).
\end{align}

Therefore it suffices to exchange the limit and integral.  To accomplish this we appeal to the
following bound from \eqref{eq:mainInequalityUniform}.
Define $f_n := \sup_{\be \in \B} \mathbb{E}_{X_n|(D_i)} ||\hat\be-\be||^2 $.  Then

\begin{align*}
|f_n| 
& := \sum_{j=1}^p f_j \\
& = \sum_{j=1}^p |\hat \be_j - \be_j|^2 \\
& \leq 
\sup_{\be \in \B} \mathbb{E}_{X_n|(D_i)} \sum_{j=1}^p\left[
\IAj\left( |Z_{nj}| + s_{nj} \right)^2
+ \IAjc|\be_j|^2 \right] \\
& \leq  \sum_{j=1}^p\left(
\frac{\ep^2}{\Delta_{nj}} + 
2s_{nj}\sqrt{\frac{\ep^2}{\Delta_{nj}}} + 
s_{nj}^2 + \parm \right)\\
& \leq  \sum_{j=1}^p \left(
\frac{\ep^2}{\Delta_{nj}}
\left(  
1 + 2\Omega_n^2 + (\Omega_n^2)^2 
\right)+ \parm \right)\\
& =  \sum_{j=1}^p \left(
\frac{\ep^2}{\Delta_{nj}}
(\Omega_n^2+1)^2 + \parm \right)
& =: \sum_{j=1}^p g_j
& =: g_n.\\
\end{align*}

Therefore, if we can exchange the limit and integral, then by the previous two proofs we
can conclude that the limit is zero.  We appeal to the following.  We say
a set of random variables $\{X_t:t\in \mathcal{T}\}$ is {\it uniformly integrable} if
\[
\lim_{x \rightarrow \infty} \sup_{t \in \mathcal{T}} \mathbb{E} |X_t| \mathbf{1}_{|X_t| > x} = 0.
\]
It holds that if $X_t \rightarrow X$ with probability one and $\{X_t:t\in\mathcal{T}\}$
is uniformly integrable, then $\mathbb{E}X_t \rightarrow \mathbb{E}X$.  Hence, we
wish to show that $f_n$ is uniformly integrable.  It holds that if each term in
the sum over $j$ is uniformly integrable, then $f_n$ is uniformly integrable as well.

Note that
\begin{align*}
\mathbb{E} |f_j| \mathbf{1}_{|f_j| > x}
& = 
x\mathbb{P}(f_j > x) + \int_x^\infty \mathbb{P}(f_j > y) dy \\
& \leq x\mathbb{P}(g_j > x) + \int_x^\infty \mathbb{P}(g_j > y) dy
\end{align*}
For large $x$, $x > T^2$ and for large $n$, $\Omega_n \asymp 1$.  Therefore, we only
need deal with the term $\epsilon^2/\Delta_{nj}$.

Using assumption (B4), continuing the above with relevant terms, and noticing that
$\sup_n f_n$ occurs at $n=1$, it follows that for $x$ large enough
\begin{align*}
x\mathbb{P}\left(\frac{1}{|D_{1j}|^2} > x\right) + \int_x^\infty \left(\frac{1}{|D_{1j}|^2} > y\right)dy 
& = 
x \left( \frac{1}{x^{\rho}} \right) + \int_x^\infty \left(\frac{1}{y^{\rho}} \right) dy \\
& = 
\left( \frac{1}{x^{\rho-1}} \right) + \int_x^\infty \left(\frac{1}{y^{\rho}} \right) dy \\
& \rightarrow 0.
\end{align*}
This allows for the exchange of integration end hence shows the desired result.

\end{proof}

\begin{proof}[Proof of Proposition \ref{prop:hatRn}]
We can expand (\ref{eq:risk}) for any $\lamVec(\mathbf{B}_n) \in \mathcal{E}$ as
\begin{equation}
R(\lamVec) := R_{\be}(\lamVec(\mathbf{B}_n)) = 
\sum_{j=1}^p \left[ (\lambda_j - 1)^2 |\be_j|^2 + \frac{\ep^2\lambda_j^2}{\Delta_{nj}} \right].
\label{eq:riskBiasVar}
\end{equation}
To form an estimator of $R$, 
we notice that $\mathbb{E}_{\be_j}(|B_{nj}|^2 - \ep^2/\Delta_{nj}) = |\be_j|^2$.
Hence,
\begin{equation}
\hat R(\lamVec) := 
\sum_{j=1}^p \left[ (\lambda_j - 1)^2\left(|B_{nj}|^2 - \sigmaSqj \right)  + 
\frac{\ep^2\lambda_j^2}{\Delta_{nj}} \right]
\label{eq:riskEstimate}
\end{equation}
is an unbiased estimate of $R(\lamVec)$.  We can make a substitution 
\begin{equation*}
\hat \psi_j := (|B_{nj}|^2 - \ep^2/\Delta_{nj})/|B_{nj}|^2,
\end{equation*}
which produces
\begin{equation}
\hat{R}(\lamVec) = 
\sum_{j=1}^p \left[ (\lambda_j - \hat\psi_j)^2|B_{nj}|^2\right] + 
\ep^2 \sum_{j=1}^p \left(\frac{\hat\psi_j}{\Delta_{nj}}\right).
\label{eq:riskEstimate2}
\end{equation}
Finally, note that the second term in $\hat{R}$ doesn't depend on $\lamVec$,
so it can be ignored for minimization purposes.  Define
\begin{equation}
\hat R_n(\lamVec) := \sum_{j=1}^p (\lambda_j - \hat\psi_j)^2|B_{nj}|^2
\end{equation}
which is proportional to $\hat R(\lamVec)$.  This is our objective function
for formulating estimators.

However, there are some natural restrictions. 
First, define $\hypSq := [0,1]^p$.   If we consider a transformed 
version of (\ref{eq:riskBiasVar})
by making the substitution $\psi_j := |\be_j|^2/(|\be_j|^2 + \Delta_{nj})$, then
\begin{equation}
R(\lamVec)  = 
\sum_{j=1}^p \left[ (\lambda_j - \psi_j)^2 \left(|\be_j|^2 + \sigmaSqj\right)
+  \ep^2 \left(\frac{\psi_j}{\Delta_{nj}}\right) \right].
\label{eq:riskBiasVar2}
\end{equation}
By inspection, the minimizer
of (\ref{eq:riskBiasVar2}) falls in $\hypSq$ as $\psi_j \in [0,1]$ for each $j$.
Hence, we cannot get a lower risk by
considering any more general sets and thus confine our attention to $\lamVec \in \hypSq$.\\
\end{proof}

\begin{proof}[Proof of Theorem \ref{thm:spectraThenAverage}]
Direct computation shows that
\[
R_1 = \min_{\lambda} \sum_{j=1}^p (1 - \lambda_j)^2 |B_j|^2 + \ep^2 \sum_{j=1}^p \frac{\lambda_j^2}{\Delta_{nj}}
\]
and
\[
R_2 = \min_{\lambda} \sum_{j=1}^p (1 - \lambda_j)^2 |B_j|^2 + \frac{\ep^2}{n} \sum_{j=1}^p \frac{\lambda_j^2}{|D_n|_j^2}.
\]
This implies that
\[
R_1 =  \sum_{j=1}^p \frac{\frac{\ep^2}{\Delta_{nj}}|\be_j|^2}{|\be_j|^2 + \frac{\ep^2}{\Delta_{nj}}}
= \sum_{j=1}^p \frac{|\be_j|^2}{\frac{\Delta_{nj}}{\ep^2}|\be_j|^2 + 1}
\]
and
\[
R_2 = \sum_{j=1}^p \frac{\frac{\ep^2}{n|D_n|_j^2}|\be_j|^2}{|\be_j|^2 + \frac{\ep^2}{n|D_n|_j^2}} 
= \sum_{j=1}^p \frac{|\be_j|^2}{\frac{n|D_n|_j^2}{\ep^2}|\be_j|^2 + 1}.
\]
Hence, the results reduces to comparing $\Delta_{nj}$ to $n |D_n|_j^2$.  Note
\[
|D_n|^2 = D_n^*D_n = \frac{1}{n^2} \sum_{i,q} D_i^* D_q
\]
therefore
\[
n |D_n|_j^2 = \frac{1}{n} \sum_{i,q} D_{ij}^*D_{qj}. 
\]
Observe
\begin{align*}
n |D_n|_j^2 - \Delta_{nj} = & \frac{1}{n} \sum_{i,q} D_{ij}^*D_{qj} - \sum_{i=1}^n |D_{ij}|^2 \\
= &  \left(\frac{1}{n} -1\right) \Delta_{nj} + \sum_{i\neq q} D_{ij}^*D_{qj} \\
\leq &  \frac{1}{n}\left( -(n - 1) \Delta_{nj} + \sum_{i\neq q} D_{ij}^*D_{qj}\right) \\
\lesssim &  -\left(n - 1\right) \sum_{i=1}^n |D_{ij}|^2 + \sum_{i\neq q} D_{ij}^*D_{qj} \\
\leq &  -\left(n - 1\right) \sum_{i=1}^n |D_{ij}|^2 + \sum_{i\neq q} (|D_{ij}|^2 + |D_{qj}|^2)/2 \\
= &  -\left(n - 1\right) \sum_{i=1}^n |D_{ij}|^2 + (n-1)\sum_{i=1}^n |D_{ij}|^2 = 0
\end{align*}
where the last inequality follows as $|D_{ij}| |D_{iq}| \leq (|D_{ij}|^2 + |D_{qj}|^2)/2$ by the arithmetic
geometric inequality.
 
\end{proof}

\bibliographystyle{imsart-nameyear}
\bibliography{references}
\end{document}